\documentclass[british]{article}
\usepackage{charter}

\usepackage[T1]{fontenc}
\usepackage[latin9]{inputenc}
\usepackage{geometry}
\usepackage{bm}
\usepackage{amsthm}
\usepackage{amsmath}
\usepackage{amssymb}
\usepackage{graphicx}
\usepackage{esint}
\usepackage[authoryear]{natbib}

\makeatletter
\theoremstyle{plain}
\newtheorem{thm}{\protect\theoremname}
  \theoremstyle{plain}
  \newtheorem{cor}[thm]{\protect\corollaryname}
  \theoremstyle{plain}
  \newtheorem{prop}[thm]{\protect\propositionname}
  \theoremstyle{plain}
  \newtheorem*{prop*}{\protect\propositionname}

\usepackage[font={small,it}]{caption}

\makeatother

\usepackage{babel}
  \providecommand{\corollaryname}{Corollary}
  \providecommand{\propositionname}{Proposition}
\providecommand{\theoremname}{Theorem}

\begin{document}
\global\long\def\ba{\bm{\alpha}}
 \global\long\def\bmu{\bm{\mu}}
 \global\long\def\bl{\bm{\lambda}}
 \global\long\def\be{\bm{\eta}}
 \global\long\def\bt{\bm{\theta}}
 \global\long\def\E{\mathbb{E}}
 \global\long\def\by{\bm{y}}
 \global\long\def\ys{\bm{y}^{\star}}
\global\long\def\xs{\mathbf{x}{}^{\star}}
 \global\long\def\bS{\bm{\Sigma}}
 \global\long\def\N{\mathcal{N}}
 \global\long\def\I{\mathbb{I}}
 \global\long\def\R{\mathbb{R}}
 \global\long\def\bb{{\bf \beta}}
\global\long\def\O{\mathcal{O}}
\global\long\def\Rc{\mathcal{R}}
\global\long\def\Vc{\mathcal{V}}
\global\long\def\bG{\bm{\Gamma}}
\global\long\def\bD{\bm{\Delta}}

\title{Fast matrix computations for functional additive models }

\author{Simon Barthelmé%
\thanks{University of Geneva, Department of Basic Neuroscience. simon.barthelme@unige.ch%
}}
\maketitle
\begin{abstract}
It is common in functional data analysis to look at a set of related
functions: a set of learning curves, a set of brain signals, a set
of spatial maps, etc. One way to express relatedness is through an
additive model, whereby each individual function $g_{i}\left(x\right)$
is assumed to be a variation around some shared mean $f(x)$. Gaussian
processes provide an elegant way of constructing such additive models,
but suffer from computational difficulties arising from the matrix
operations that need to be performed. Recently Heersink \& Furrer
have shown that functional additive model give rise to covariance
matrices that have a specific form they called \emph{quasi-Kronecker}
(QK), whose inverses are relatively tractable. We show that under
additional assumptions the two-level additive model leads to a class
of matrices we call \emph{restricted quasi-Kronecker }(rQK)\emph{,
}which enjoy many interesting properties. In particular, we formulate
matrix factorisations whose complexity scales only linearly in the
number of functions in latent field, an enormous improvement over
the cubic scaling of naïve approaches. We describe how to leverage
the properties of rQK matrices for inference in Latent Gaussian Models.
\end{abstract}

\section{Introduction\label{sec:Introduction}}

One of the most frequent scenarios in functional data analysis is
that one has a group of related functions (for example, learning curves,
growth curves, EEG traces, etc.) and the goal is to describe what
these functions have in common and how they vary \citep{RamsaySilverman:FunctionalDataAnalysis,Behseta:HierarchModelsVariabFunctions,Cheng:BayesianRegistFunctionsCurves}.
In that context functional additive models are very natural, and the
simplest is the following two-level model:

\begin{equation}
g_{i}\left(x\right)=f\left(x\right)+d_{i}\left(x\right)\label{eq:two-level-model}
\end{equation}

where functions $g_{1}\left(x\right),\ldots,g_{m}\left(x\right)$
are $m$ individual functions (e.g. learning curves for $m$ different
individuals), $f\left(x\right)$ is a mean function and $d_{1}\left(x\right)\ldots d_{n}\left(x\right)$
describe individual differences. 

The two-level model appears on its own, or as an essential building
block in more complex models, which may include several hierarchical
levels (as in functional ANOVA, \citealp{RamsaySilverman:FunctionalDataAnalysis,KaufmanSain:BayesianFuncANOVA,Sain:fANOVAandRegClimateExperiments}),
time shifts \citep{Cheng:BayesianRegistFunctionsCurves,KneipRamsay:CombiningRegistrationFitting,TelescaInoue:BayesHierarchCurveReg},
or non-Gaussian observations \citep{Behseta:HierarchModelsVariabFunctions}.
Functional PCA \citep{RamsaySilverman:FunctionalDataAnalysis}, which
seeks to characterise the distribution of the difference functions
$d_{i}$, is a closely related variant . In a Bayesian framework the
two-level model can be expressed as a particular kind of Gaussian
process prior \citep{KaufmanSain:BayesianFuncANOVA}. The goal of
this paper is to show that the covariance matrices that arise in the
two-level model have a form that lends itself to very efficient computation.
We call these matrices \emph{restricted quasi-Kronecker}, after \citet{HeersinkFurrer:MoorePenroseInversesQuasiKronecker}. 

The paper is structured as follows. We first give some background
material on Gaussian processes and latent Gaussian models \citep{Rue:INLA},
and describe the particular covariance matrices that obtain in functional
additive models. These matrices are \emph{quasi-Kronecker }(QK) or
\emph{restricted quasi-Kronecker }(rQK)\emph{, }and in the next section
we prove some theoretical results that follow from the \emph{block-rotated
}form of rQK matrices. In particular, rQK matrices have very efficient
factorisations, and we detail a range of useful computations based
on these factorisations. In the following section we apply our results
to Gaussian data (joint smoothing), and show that marginal likelihoods
and their derivatives are highly tractable. The other application
we highlight concerns the modelling of spike trains, which leads to
a latent Gaussian model with Poisson likelihood. We show that the
Hessian matrices of the log-posterior in the two-level model are quasi-Kronecker,
which makes the Laplace approximation and its derivative tractable.
This leads to efficient approximate inference methods for large-scale
functional additive models.

\subsection{Notation}

The Kronecker product of $\mathbf{A}$ and $\mathbf{B}$ is denoted
$\mathbf{A}\otimes\mathbf{B}$, the all-one vector of length $n$
$\mathbf{e}_{n}$ or $\mathbf{e}$ if obvious from the context. Throughout
$m$ is used for the number of individual functions in the functional
additive model, and $n$ for the number of grid points. 

In what follows we will use the following two properties of the Kronecker
product and vec operators \citep{PetersenPedersen:MatrixCookbook}:

\begin{equation}
\left(\mathbf{A}\otimes\mathbf{B}\right)\left(\mathbf{C}\otimes\mathbf{D}\right)=\left(\mathbf{AC}\otimes\mathbf{BD}\right)\label{eq:kron-prod}
\end{equation}

for compatible matrices, and 

\begin{equation}
\mbox{vec}\left(\mathbf{AXB}\right)=\left(\mathbf{B}^{t}\otimes\mathbf{A}\right)\mbox{vec}\left(\mathbf{X}\right)\label{eq:vec-kron}
\end{equation}

where the $\mbox{vec}\left(\mathbf{X}\right)$ operator stacks the
columns of $\mathbf{X}$ vertically.

\subsection{Gaussian process and latent Gaussian models\label{sub:Gaussian-process-intro}}

We give here a brief description of Gaussian processes (GPs), a much
more detailed treatment is available in \citet{RasmussenWilliamsGP}.
A GP with mean 0 and covariance function $k(x,x')$ is a distribution
on the space of functions of some input space $\mathcal{X}$ into
$\mathbb{R}$, such that for every set of sampling points $\left\{ x_{1},\ldots,x_{n}\right\} $,
the sampled values $f(x_{1}),\ldots,f(x_{n})$ follow a multivariate
Gaussian distribution

\begin{eqnarray*}
f(x_{1}),\ldots,f(x_{n}) & \sim & \N\left(0,\mathbf{K}\right)\\
\mathbf{K}_{ij} & = & k(x_{i},x_{j})
\end{eqnarray*}

As a shorthand for the above, we will use the notation $f\sim GP\left(0,\kappa\right)$
in what follows. The covariance function usually expresses the idea
that we expect the values of $f$ at $x_{i}$ and $x_{j}$ to be close
if $x_{i}$ and $x_{j}$ are themselves close. Often, covariance functions
only depend on the distance between two points, so that $k\left(x_{i},x_{j}\right)=\kappa\left(d\left(x_{i},x_{j}\right)\right)$,
with $d\left(x_{i},x_{j}\right)$ some distance function (usually
Euclidean). 

In most actual cases the covariance function is not known and involves
two or more hyperparameters, which we will call $\bt$. In machine
learning the squared-exponential covariance function is especially
popular:

\begin{equation}
\kappa_{\bt}\left(d\right)=\exp\left(-e^{-\theta_{1}}\frac{d^{2}}{2}+\theta_{2}\right)\label{eq:covfun-squared-exp}
\end{equation}

where $\theta_{1}$ is a (log) length-scale parameter and $\theta_{2}$
controls the marginal variance of the process. For reasons of numerical
stability we prefer to use a Matern 5/2 covariance function, which
generates rougher functions but leads to covariance matrices that
are better behaved. The Matern 5/2 covariance function has the following
form:

\begin{eqnarray}
\kappa_{\bt}\left(u\right) & = & \left(1+u+\frac{u^{2}}{3}\right)\exp\left(-u+\theta_{2}\right)\label{eq:covfun-matern52}\\
u & = & \sqrt{5}d\times e^{-\theta_{1}}\nonumber 
\end{eqnarray}

given in \citet{RasmussenWilliamsGP}, page 85. The parameters $\theta_{1}$
and $\theta_{2}$ play similar roles in this formulation, controlling
length-scale and marginal variance.

Gaussian processes can be used to formulate priors for Latent Gaussian
Models \citep{Rue:INLA}. The term ``Latent Gaussian Model'' describes
a very large class of statistical models, one that encompasses for
example all Generalised Linear Models (with Gaussian priors on the
coefficients), Generalised Mixed Models, and Generalised Additive
Models. The two main components are (a) a latent Gaussian field, $\mathbf{z}=\left[z_{1}\ldots z_{n}\right]$,
with prior $\mathbf{z}\sim\N\left(0,\bS_{\bt}\right)$ and (b) independently
distributed data $y_{i}\sim f(y|z_{i})$, which depend on the corresponding
value of the latent field. In regression models $f(y|z_{i})$ is Gaussian,
but other common distributions include a logit model (for binary $y$)
and exponential-Poisson (as in section \ref{sub:Spike-train-data}).
The results we give here apply to all Latent Gaussian Models where
the latent field has the structure of a two-level functional additive
model, which we detail next. 

\begin{figure}

\begin{centering}
\includegraphics[width=8cm]{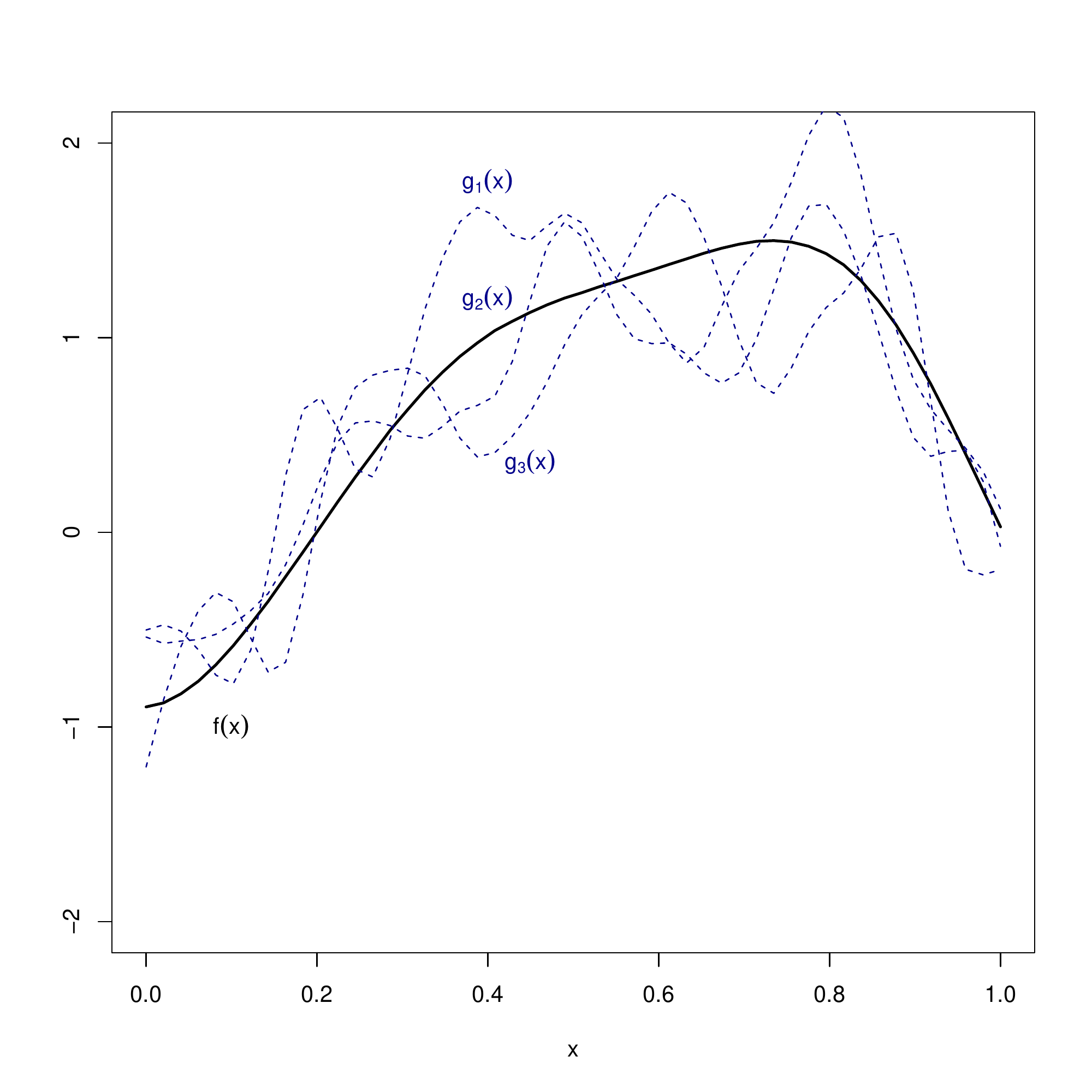}
\par\end{centering}

\caption{The functional additive model (eq. \ref{eq:two-level-model-GP}) produces
sets of related functions. Here we show a latent mean function ($f(x)$),
sampled from a Gaussian process. The individual functions $g_{1},g_{2},g_{3}$
are variations on $f$ sampled from a Gaussian process with mean function
$f(x)$.\label{fig:illus-functional-additive-model}}

\end{figure}

\subsection{Quasi-Kronecker structures in functional additive models\label{sub:Quasi-Kronecker-structures}}

In a Gaussian process context the two-level functional additive model
translates into the following model:

\begin{eqnarray}
f & \sim & GP\left(0,\kappa\right)\nonumber \\
g_{1}\ldots g_{m}\vert f & \sim & GP\left(f,\kappa'\right)\label{eq:two-level-model-GP}
\end{eqnarray}

where individual $g_{i}$'s are conditionally independent given $f$
(they are however marginally dependent). 

Fig. \ref{fig:illus-functional-additive-model} gives an example of
a draw from such a model. Eq. \ref{eq:two-level-model-GP} expresses
the model as a prior over functions, but if we have a set of sampling
locations $x_{1}\ldots x_{n}$ the Gaussian process assumption implies
a multivariate Gaussian model for the latent field at the sample locations.
We assume throughout that there are $n\times m$ observations, corresponding
to $m$ realisations of a latent process observed on a grid of size
$n$, with grid points $x_{1}\ldots x_{n}$. The grid may be irregular
(i.e. $x_{t+1}-x_{t}$ needs not equal a constant), but we need the
grid to be constant over individual functions, since otherwise the
covariance matrix does not have the requisite structure. Denote by
$\mathbf{g}_{i}$ the vector of sampled values from $g_{i}(x)$: $\mathbf{g}_{i}=\left[g_{1}\left(x_{1}\right)\ldots g_{n}\left(x_{n}\right)\right]^{t}$.
The latent field $G\in\R^{n\times m}$ is formed of the concatenation
of $\mathbf{g}_{1}\ldots\mathbf{g}_{m}$, and we may think of it either
as a $n\times m$ matrix $\mathbf{G}$ or a single vector $\mbox{vec}\left(\mathbf{G}\right)$.
Following the LGM framework, we assume that the data are (possibly
non-linear, non-Gaussian) observations based on the latent values,
so that $y_{ij}\sim f(y|g_{ij})$.

Under a particular grid, sampled values from the mean function $\mathbf{f}$
have distribution $\mathbf{f}\sim\N\left(0,\mathbf{K}\right)$, and
$\mathbf{g}_{i}$ has conditional distribution $\mathbf{g}_{i}|\mathbf{f}\sim\N\left(\mathbf{f},\mathbf{A}\right)$.
Draws from the latent field $\mathbf{g}$ can therefore be obtained
by the following transformation:

\begin{eqnarray}
\mathbf{g} & = & \left[\begin{array}{ccccc}
\mathbf{I} & \mathbf{I}\\
\mathbf{I} &  & \mathbf{I}\\
\vdots &  &  & \ddots\\
\mathbf{I} &  &  &  & \mathbf{I}
\end{array}\right]\left[\begin{array}{c}
\mathbf{f}\\
\mathbf{d}_{1}\\
\vdots\\
\mathbf{d}_{m}
\end{array}\right]\nonumber \\
 & = & \mathbf{M}\left[\begin{array}{c}
\mathbf{f}\\
\mathbf{d}_{1}\\
\vdots\\
\mathbf{d}_{m}
\end{array}\right]\label{eq:generating-latent-field}
\end{eqnarray}

where $\mathbf{f}\sim\N\left(0,\mathbf{K}\right)$ and $\mathbf{d}_{i}\sim\N\left(0,\mathbf{A}\right)$.
We can use eq. \ref{eq:generating-latent-field} to find the marginal
distribution of $\mathbf{g}$ (unconditional on $\mathbf{f}$). Using
standard properties of Gaussian random variables we find:

\begin{eqnarray}
\mathbf{g} & \sim & \N\left(0,\bS\right)\nonumber \\
\bS & = & \mathbf{M}\left[\begin{array}{cccc}
\mathbf{K}\\
 & \mathbf{A}\\
 &  & \ldots\\
 &  &  & \mathbf{A}
\end{array}\right]\mathbf{M}^{t}\nonumber \\
 & = & \left[\begin{array}{cccc}
\mathbf{A}+\mathbf{K} & \mathbf{K} & \cdots & \mathbf{K}\\
\mathbf{K} & \mathbf{A}+\mathbf{K} & \cdots & \mathbf{K}\\
\vdots & \vdots & \ddots & \vdots\\
\mathbf{K} & \mathbf{K} & \cdots & \mathbf{A}+\mathbf{K}
\end{array}\right]\label{eq:marginal-g}
\end{eqnarray}

$\bS$ is a dense matrix, whose dimensions ($nm\times nm$) would
seem to preclude direct inference, since the costs of factorising
such a matrix would be too high with numbers as low as 100 grid points
and 20 observed functions. \citet{HeersinkFurrer:MoorePenroseInversesQuasiKronecker}
have shown that the matrix inverse of $\bS$ is actually surprising
tractable, and in the following section we summarise and extend their
results.

\section{Some properties of restricted quasi-Kronecker matrices\label{sec:Some-properties-rQK}}

Quasi-Kronecker (QK) matrices are introduced in \citet{HeersinkFurrer:MoorePenroseInversesQuasiKronecker}
and have the following form:

\begin{equation}
\bS=\mbox{bdiag}\left(\mathbf{A}_{1}\ldots\mathbf{A}_{m}\right)+\mathbf{uv}^{t}\otimes\mathbf{K}\label{eq:definition-QK}
\end{equation}

We focus mostly on the restricted case (rQK), i.e. matrices of the
form 

\begin{equation}
\bS=\mbox{bdiag}\left(\mathbf{A}\ldots\mathbf{A}\right)+\mathbf{ee}^{t}\otimes\mathbf{K}=\mathbf{I}_{m}\otimes\mathbf{A}+\mathbf{ee}^{t}\otimes\mathbf{K}=\mbox{rQK}\left(\mathbf{A},\mathbf{K}\right)\label{eq:rQK-definition}
\end{equation}

which arise in the functional additive model described just above
(eq. \ref{eq:marginal-g}).

QK and especially rQK matrices have a number of properties that make
them much more tractable than dense, unstructured matrices of the
same size.

\subsection{The general (unrestricted) case}

The cost of matrix-vector products with QK matrices is $\O\left(n^{2}m\right)$,
as compared to $\O\left(n^{2}m^{2}\right)$ in the general case. This
follows directly from the definition (eq. \ref{eq:definition-QK}),
as we need only perform $\O\left(m\right)$ operations of complexity
$\O\left(n^{2}\right)$.

\citet{HeersinkFurrer:MoorePenroseInversesQuasiKronecker} showed
that the inverse and pseudo-inverse of QK matrices is also tractable.
For the inverse the following formula applies:

\begin{eqnarray}
\left(\mbox{bdiag}\left(\mathbf{A}_{1}\ldots\mathbf{A}_{m}\right)+\mathbf{uv}^{t}\otimes\mathbf{K}\right)^{-1} & = & \mbox{bdiag}\left(\mathbf{A}_{i}^{-1}\right)-\mbox{bdiag}\left(\mathbf{A}_{i}^{-1}\right)\left(\mathbf{u}\otimes\mathbf{L}_{\mathbf{k}}\right)\mathbf{P}^{-1}\left(\mathbf{v}\otimes\mathbf{L}_{\mathbf{k}}\right)^{t}\mbox{bdiag}\left(\mathbf{A}_{i}^{-1}\right)\label{eq:QK-inverse}\\
\mathbf{P} & = & \left(\mathbf{I}_{m}+\left(\mathbf{v}\otimes\mathbf{L}_{\mathbf{k}}\right)^{t}\mbox{bdiag}\left(\mathbf{A}_{i}^{-1}\right)\left(\mathbf{u}\otimes\mathbf{L}_{k}\right)\right)\nonumber 
\end{eqnarray}

where $\mathbf{K=}\mathbf{L}_{\mathbf{k}}\mathbf{L}_{\mathbf{k}}^{t}$.
The result can be proved using the Sherman-Woodbury-Morrison formula,
by writing $\mathbf{uv}^{t}\otimes\mathbf{K}=\left(\mathbf{u}\otimes\mathbf{L}_{\mathbf{k}}\right)\left(\mathbf{v}\otimes\mathbf{L}_{\mathbf{k}}\right)^{t}$,
which follows from eq. (\ref{eq:kron-prod}). It implies that a $mn\times mn$
QK matrix can be inverted in $\O\left(n^{3}m+m^{3}\right)$ operations,
as opposed to $\O\left(m^{3}n^{3}\right)$ in the general case. A
similar formula holds for the determinant:

\begin{equation}
\mbox{det}\left(\mbox{bdiag}\left(\mathbf{A}_{1}\ldots\mathbf{A}_{m}\right)+\mathbf{uv}^{t}\otimes\mathbf{K}\right)=\mbox{det}\left(\mbox{bdiag}\left(\mathbf{A}_{1}\ldots\mathbf{A}_{m}\right)\right)\det\left(\mathbf{P}\right)=\prod\det\left(\mathbf{A}_{i}\right)\det\left(\mathbf{P}\right)\label{eq:determinant-QK}
\end{equation}

where $\mathbf{P}$ is as in equation \ref{eq:QK-inverse}. The result
implies that determinants of QK matrices can be computed in $\O\left(n^{3}m+m^{3}\right)$
operations ($\O\left(m^{3}n^{3}\right)$ in the general case). We
show in section \ref{sub:Spike-train-data} that it can be used to
speed up Laplace approximations in latent Gaussian models \citep{Rue:INLA}.

\subsection{The restricted case}

rQK matrices have some additional properties not enjoyed by general
QK matrices. For example, the product of two rQK matrices is another
rQK matrix:

\begin{equation}
\mbox{rQK}\left(\mathbf{A}_{1},\mathbf{K}_{1}\right)\mbox{rQK}\left(\mathbf{A}_{2},\mathbf{K}_{2}\right)=\mbox{rQK}\left(\mathbf{A}_{1}\mathbf{A}_{2},\mathbf{A}_{1}\mathbf{K}_{2}+\mathbf{K}_{1}\mathbf{A}_{2}+m\mathbf{K}_{1}\mathbf{K}_{2}\right)\label{eq:prod-two-rQKs}
\end{equation}

We use this property below to compute the gradient of the marginal
likelihood in Gaussian models (section \ref{sub:gaussian-marginal-likelihood}).
In addition, several other properties follow from the block-rotated
form of rQK matrices: for example, the inverse of a rQK matrix is
another rQK matrix, and the eigenvalue decomposition has a very  structure.
We prove these results next.

\subsection{Block-rotated form of rQK matrices\label{sub:Block-rotated-form}}

In this section we first establish that rQK matrices can be \emph{block-rotated
}into a block-diagonal form. From this result the eigendecomposition
follows immediately, and so does a Cholesky-based square root. These
decompositions need not be formed explicitly, and we will see that
their cost scales as $\mathcal{O}\left(n^{2}\right)$ in storage and
$\mathcal{O}\left(mn^{3}\right)$ in time. The latter is linear in
$m$, compared to cubic for naïve algorithms.

We begin by defining block rotations, which are straightforward extensions
of block permutations. A block permutation can be written using the
Kronecker product $\mathbf{P}\otimes\mathbf{I}_{n}$ of a $m\times m$
permutation matrix $\mathbf{P}$ and the $n\times n$ identity matrix.
Applied to a matrix \textbf{$\mathbf{M}$}, $\left(\mathbf{P}\otimes\mathbf{I}_{n}\right)\mathbf{M}\left(\mathbf{P}^{t}\otimes\mathbf{I}_{n}\right)$
will permute blocks of size $n\times n$. Block rotations are defined
in a similar way, through the Kronecker product $\mathbf{R}=\mathbf{B}\otimes\mathbf{I}_{n}$
of a $m\times m$ rotation matrix $\mathbf{B}$ and the $n\times n$
identity matrix. A block rotation matrix is an orthogonal matrix,
since:

\begin{eqnarray*}
\left(\mathbf{B}\otimes\mathbf{I}_{n}\right)^{t}\left(\mathbf{B}\otimes\mathbf{I}_{n}\right) & = & \left(\mathbf{B}^{t}\otimes\mathbf{I}_{n}\right)\left(\mathbf{B}\otimes\mathbf{I}_{n}\right)\\
 & = & \left(\mathbf{B}^{t}\mathbf{B}\otimes\mathbf{I}_{n}\right)\\
 & = & \mathbf{I}_{mn}
\end{eqnarray*}

where the second line follows from eq. (\ref{eq:kron-prod}). 

Our main result is the following:
\begin{thm}
There exists an orthogonal matrix $\mathbf{B}$ such that for all
$n\times n$ matrices $\mathbf{A},\mathbf{K}$, the matrix $\bS=\mbox{rQK}\left(\mathbf{A},\mathbf{K}\right)$
can be expressed as:

\[
\bS=\left(\mathbf{B}\otimes\mathbf{I}_{n}\right)^{t}\left(\begin{array}{cccc}
\left(\mathbf{A}+m\mathbf{K}\right)\\
 & \mathbf{A}\\
 &  & \ddots\\
 &  &  & \mathbf{A}
\end{array}\right)\left(\mathbf{B}\otimes\mathbf{I}_{n}\right)
\]

\end{thm}
Since $\mathbf{R}=\left(\mathbf{B}\otimes\mathbf{I}_{n}\right)$ is
an orthogonal matrix, and the inner matrix is block-diagonal, the
inverse, eigendecomposition, and a Cholesky-based square root of $\bS$
all follow easily. 
\begin{proof}
We use the following ansatz: the $m\times m$ matrix $\mathbf{B}$
is an orthogonal matrix whose first row is a scaled version of the
all-ones vector\\
\[
\mathbf{B}=\left[\begin{array}{c}
\mathbf{e}^{t}/\sqrt{m}\\
\mathbf{L}
\end{array}\right]
\]

and $\mathbf{L}$ is chosen such that $\mathbf{B}^{t}\mathbf{B}=\mathbf{I}$,
so that the rows of \textbf{L }will be orthogonal to $\mathbf{e}$,
and $\mathbf{B}\mathbf{e}=\left[\begin{array}{cccc}
m/\sqrt{m} & 0 & \ldots & 0\end{array}\right]$. Note that \textbf{L} is not unique, but a matrix that verifies the
condition can always be found by the Gram-Schmidt process \citep{GolubVanLoan:MatrixComputations}
(although we will see below that a better option is available). 

As noted above $\bS=\left(\mathbf{I}_{m}\otimes\mathbf{A}\right)+\left(\mathbf{ee^{t}}\otimes\mathbf{K}\right)$.
We left-multiply by $\mathbf{R}=\left(\mathbf{B}\otimes\mathbf{I}_{n}\right)$
and right-multiply by $\mathbf{R}^{t}$:

\begin{eqnarray}
\mathbf{R}\bS\mathbf{R}^{t} & = & \left(\mathbf{B}\otimes\mathbf{I}_{n}\right)\left(\mathbf{I}_{m}\otimes\mathbf{A}\right)\left(\mathbf{B}^{t}\otimes\mathbf{I}_{n}\right)+\left(\mathbf{B}\otimes\mathbf{I}_{n}\right)\left(\mathbf{ee}^{t}\otimes\mathbf{K}\right)\left(\mathbf{B}^{t}\otimes\mathbf{I}_{n}\right)\nonumber \\
 & = & \left(\mathbf{B}\otimes\mathbf{A}\right)\left(\mathbf{B}^{t}\otimes\mathbf{I}_{n}\right)+\left(\mathbf{B}\mathbf{ee}^{t}\otimes\mathbf{K}\right)\left(\mathbf{B}^{t}\otimes\mathbf{I}_{n}\right)\nonumber \\
 & = & \left(\mathbf{B}\mathbf{B}^{t}\otimes\mathbf{A}\right)+\left(\mathbf{B}\mathbf{ee}^{t}\mathbf{B}^{t}\otimes\mathbf{K}\right)\nonumber \\
 & = & \left(\mathbf{I}_{m}\otimes\mathbf{A}\right)+\begin{array}{ccc}
m\mathbf{K} & 0 & \ldots\\
0 & 0 & \ldots\\
\vdots & \vdots & \ddots
\end{array}\label{eq:sigma-rotated}
\end{eqnarray}

Left-multiplying by $\mathbf{R}^{t}$ and right-multiplying by $\mathbf{R}$
completes the proof.
\end{proof}

\subsection{Inverse, factorised form, and eigenvalue decomposition\label{sub:Inverse-factorised-form-eigenvalue}}

A number of interesting properties follow directly from the block-rotated
form given by Theorem 1. 
\begin{cor}
The inverse of an rQK matrix is also an rQK matrix. Specifically, 

$\mbox{rQK}\left(\mathbf{A},\mathbf{K}\right)^{-1}=\mbox{rQK}\left(\mathbf{A}^{-1},\frac{1}{m}\left(\left(\mathbf{A}+m\mathbf{K}\right)^{-1}-\mathbf{A}^{-1}\right)\right)$\end{cor}
\begin{proof}
This result can also be derived from the Furrer-Heersink formula (eq.
\ref{eq:QK-inverse}) but follows naturally from Theorem 1. Assuming
that $\mathbf{A}+m\mathbf{K}$ and $\mathbf{K}$ have full rank, then
the inverse of $\bS$ exists and equals:

\[
\bS^{-1}=\left(\mathbf{B}\otimes\mathbf{I}_{n}\right)^{t}\left(\begin{array}{cccc}
\left(\mathbf{A}+m\mathbf{K}\right)^{-1}\\
 & \mathbf{A}^{-1}\\
 &  & \ddots\\
 &  &  & \mathbf{A}^{-1}
\end{array}\right)\left(\mathbf{B}\otimes\mathbf{I}_{n}\right)
\]

For $\bS^{-1}$ to be an rQK matrix we need to find matrices $\mathbf{A}'$
and $\mathbf{K}'$ such that $\mathbf{A}'+m\mathbf{K}'=\left(\mathbf{A}+m\mathbf{K}\right)^{-1}$,
and $\mathbf{A}'=\mathbf{A}^{-1}$. This implies 
\begin{eqnarray*}
\mathbf{A}+m\mathbf{K}' & = & \left(\mathbf{A}+m\mathbf{K}\right)^{-1}\\
\mathbf{K}' & = & \frac{1}{m}\left(\left(\mathbf{A}+m\mathbf{K}\right)^{-1}-\mathbf{A}^{-1}\right)
\end{eqnarray*}

which requires the existence of $\left(\mathbf{A}+m\mathbf{K}\right)^{-1}$
and $\mathbf{A}^{-1}$, both conditions being verified if $\bS$ is
invertible. 
\end{proof}
An important consequence of this result is that Hessian matrices in
latent Gaussian models with prior covariance $\bS$ turn out to be
quasi-Kronecker (section \ref{sub:Spike-train-data}), so that Laplace
approximations can be computed in $\O\left(mn^{3}\right)$ time. 
\begin{cor}
The eigenvalue decomposition of $\bS=\mbox{rQK}\left(\mathbf{A},\mathbf{K}\right)$
is given by

\[
\bS=\mathbf{R}^{t}\mathbf{N}^{t}\mathbf{D}\mathbf{N}\mathbf{R}
\]

where $\mathbf{R}$ is as in theorem 1, $\mathbf{N}=\mbox{bdiag}\left(\mbox{ev}\left(\mathbf{A}+m\mathbf{K}\right),\mbox{ev}\left(\mathbf{A}\right),\ldots,\mbox{ev}\left(\mathbf{A}\right)\right)$
and $\mathbf{D}=\mbox{diag}\left(\lambda\left(\mathbf{A}+m\mathbf{K}\right),\lambda\left(\mathbf{A}\right),\ldots,\lambda\left(\mathbf{A}\right)\right)$. \end{cor}
\begin{proof}
Since $\left(\begin{array}{cccc}
\left(\mathbf{A}+m\mathbf{K}\right)\\
 & \mathbf{A}\\
 &  & \ddots\\
 &  &  & \mathbf{A}
\end{array}\right)$ is block-diagonal its eigendecomposition into $\mathbf{N}^{t}\mathbf{D}\mathbf{N}$
is straightforward. The matrix $\mathbf{NR}$ is orthogonal: $\mathbf{R}^{t}\mathbf{N}^{t}\mathbf{N}\mathbf{R}=\mathbf{I}$,
and since $\bS=\mathbf{R}^{t}\mathbf{N}^{t}\mathbf{D}\mathbf{N}\mathbf{R}$
with diagonal $\mathbf{D}$ we have its eigenvalue decomposition.\end{proof}
\begin{cor}
A set of square root factorisations of \textbf{$\bS$} is given by
$\bS=\mathbf{G}^{t}\mathbf{G}$, with 
\[
\mathbf{G}=\mathbf{\left(\begin{array}{cccc}
\mathbf{U}\\
 & \mathbf{V}\\
 &  & \ddots\\
 &  &  & \mathbf{V}
\end{array}\right)}\mathbf{R}
\]

where $\mathbf{U}^{t}\mathbf{U}=\mathbf{A}+m\mathbf{K}$ and $\mathbf{V}^{t}\mathbf{V}=\mathbf{A}$. 
\end{cor}
This corollary is an entirely straightforward consequence of the theorem,
but note that $\mathbf{G}$ has special form, as the product of a
block-diagonal matrix and a block-permutation matrix. In the next
section we outline an algorithm that takes advantage of these properties
to compute the factorisation with initial cost $\O\left(n^{3}\right)$
and further cost $\O\left(mn^{2}\right)$ per MVP.

\subsection{A fast algorithm for square roots of $\bS$\label{sub:A-fast-algorithm-square-roots}}

Our fast algorithm is based on Corollary 4. Given a matrix square
root $\bS=\mathbf{G}^{t}\mathbf{G}$, the typical operations used
in statistical computation are the following:
\begin{itemize}
\item Correlating transform: given $\mathbf{z}\sim\N\left(0,\mathbf{I}\right)$
the vector $\mathbf{x}=\mathbf{G}^{t}\mathbf{z}$ is distributed according
to $\mathbf{x}\sim\N\left(0,\bS\right)$. The correlating transform
(MVP with $\mathbf{G}^{t}$) takes a IID vector and gives it the right
correlation structure. 
\item ``Whitening'' transform: given $\mathbf{x}\sim\N\left(0,\bS\right)$
the vector $\mathbf{z}=\mathbf{G}^{-t}\mathbf{x}$ is distributed
according to $\mathbf{z}\sim\N\left(0,\mathbf{I}\right)$. This is
the inverse operation of the correlating transform and produces white
noise from a correlated vector. 
\item Evaluation of $p\left(\mathbf{x}\vert\bmu,\bS\right)=\left(2\pi\right)^{-mn/2}\exp\left(-\frac{1}{2}\left(\mathbf{x}-\bmu\right)^{t}\bS^{-1}\left(\mathbf{x}-\bmu\right)-\log\det\bS^{-1}\right)$
. This can be done using the whitening transform but in addition the
log-determinant of $\bS$ is needed. We give a formula below. 
\end{itemize}
Our algorithm is based on the square roots $\mathbf{U}^{t}\mathbf{U}=\mathbf{A}+m\mathbf{K}$
and $\mathbf{V}^{t}\mathbf{V}=\mathbf{A}$, which can be obtained
from the Cholesky or eigendecompositions. In most cases the Cholesky
version is faster but the eigendecomposition has advantages in certain
contexts. We outline a generic algorithm here, the only practical
difference being in solving the linear systems $\mathbf{U}^{-t}\mathbf{x}$
and $\mathbf{V}^{-t}\mathbf{x}$, which should be done via forward
or back substitution if $\mathbf{U}$ and $\mathbf{V}$ are Cholesky
factors \citep{GolubVanLoan:MatrixComputations}.

\subsubsection{Computing the correlating and whitening transforms }

Corrolary 4 tells us that the square root of $\bS$, $\mathbf{G}^{t}$
equals: 

\begin{equation}
\mathbf{G}^{t}=\left(\mathbf{B}^{t}\otimes\mathbf{I}_{n}\right)\left(\begin{array}{cccc}
\mathbf{U}^{t}\\
 & \mathbf{V}^{t}\\
 &  & \ddots\\
 &  &  & \mathbf{V}^{t}
\end{array}\right)\label{eq:correlating-trans}
\end{equation}

We first show how to compute the correlating transform. Let $\mathbf{z}=\mbox{vec}\left(\left[\mathbf{z}_{1}\ldots\mathbf{z}_{m}\right]\right)=\mbox{vec}\left(\mathbf{Z}\right)$.
Using eq. \ref{eq:vec-kron}, we can rewrite $\mathbf{G}^{t}\mathbf{z}$
as:

\begin{eqnarray*}
\mathbf{G}^{t}\mathbf{z} & = & \left(\mathbf{B}^{t}\otimes\mathbf{I}_{n}\right)\left(\begin{array}{cccc}
\mathbf{U}^{t}\\
 & \mathbf{V}^{t}\\
 &  & \ddots\\
 &  &  & \mathbf{V}^{t}
\end{array}\right)\mbox{vec}\left(\mathbf{Z}\right)\\
 & = & \left(\mathbf{B}^{t}\otimes\mathbf{I}_{n}\right)\mbox{vec}\left(\left[\begin{array}{cccc}
\mathbf{U}\mathbf{z}_{1}^{t} & \mathbf{V}^{t}\mathbf{z}_{2} & \ldots & \mathbf{V}^{t}\mathbf{z}_{m}\end{array}\right]\right)\\
 & = & \left(\mathbf{B}^{t}\otimes\mathbf{I}_{n}\right)\mbox{vec}\left(\mathbf{Y}\right)\\
 & = & \mbox{vec}\left(\mathbf{Y}\mathbf{B}\right)
\end{eqnarray*}

Computing the correlating transform entails $m$ MVP products with
the square roots$\begin{array}{cccc}
\mathbf{U}\mathbf{z}_{1}^{t} & \mathbf{V}^{t}\mathbf{z}_{2} & \ldots & \mathbf{V}^{t}\mathbf{z}_{m}\end{array}$ ($\O\left(n^{2}\right)$ in the general case) and $m$ MVP products
with $\mathbf{B}$. The latter $m$ products can be done in $\O\left(mn\right)$
time, as we show below, meaning that the whole operation has total
cost $\O\left(mn^{2}\right)$.

The whitening transform can be computed in a similar way:
\begin{equation}
\mathbf{G}^{-t}=\left(\begin{array}{cccc}
\mathbf{U}^{-t}\\
 & \mathbf{V}^{-t}\\
 &  & \ddots\\
 &  &  & \mathbf{V}^{-t}
\end{array}\right)\left(\mathbf{B}\otimes\mathbf{I}_{n}\right)\label{eq:whitening-transform}
\end{equation}

so that for $\mathbf{x}=\mbox{vec}\left(\left[\mathbf{x}_{1}\ldots\mathbf{x}_{m}\right]\right)=\mbox{vec}\left(\mathbf{X}\right)$
\begin{eqnarray*}
\mathbf{G}^{-t}\mathbf{x} & = & \left(\begin{array}{cccc}
\mathbf{U}^{-t}\\
 & \mathbf{V}^{-t}\\
 &  & \ddots\\
 &  &  & \mathbf{V}^{-t}
\end{array}\right)\left(\mathbf{B}\otimes\mathbf{I}_{n}\right)\mbox{vec}\left(\mathbf{X}\right)\\
 & = & \mbox{bdiag}\left(\mathbf{U}^{-t},\mathbf{V}^{-t},\ldots,\mathbf{V}^{-t}\right)\mbox{vec}\left(\mathbf{X}\mathbf{B}^{t}\right)\\
 & = & \mbox{bdiag}\left(\mathbf{U}^{-t},\mathbf{V}^{-t},\ldots,\mathbf{V}^{-t}\right)\mbox{vec}\left(\mathbf{Y}\right)\\
 & = & \mbox{vec}\left(\left[\begin{array}{cccc}
\mathbf{U}^{-t}\mathbf{y}_{1} & \mathbf{V}^{-t}\mathbf{y}_{2} & \ldots & \mathbf{V}^{-t}\mathbf{y}_{m}\end{array}\right]\right)
\end{eqnarray*}

and since matrix solves of the form $\mathbf{U}^{-t}\mathbf{y}$ and
$\mathbf{V}^{-t}\mathbf{y}$ have $\O\left(n^{2}\right)$ cost for
Cholesky or eigenfactors, the entire operation has the same cost as
the correlating transform ($\O\left(mn^{2}\right)$).

\subsubsection{Fast rotations\label{sub:Fast-rotations}}

The whitening and correlating transforms above involve MVPs with a
$m\times m$ orthogonal matrix $\mathbf{B}=\left[\begin{array}{c}
\mathbf{e}^{t}/\sqrt{m}\\
\mathbf{L}
\end{array}\right]$, where $\mathbf{L}$ is such that $\mathbf{B}$ is orthonormal. Construction
of $\mathbf{L}$ via the Gram-Schmidt process would have initial cost
$\mathcal{O}\left(m^{3}\right)$ and MVPs with $\mathbf{B}$ $\O\left(m^{2}\right)$.
There is however a particular choice for $\mathbf{L}$ that brings
these costs down to $\O\left(1\right)$ and $\mathcal{O}\left(m\right)$:
\begin{prop}
The matrix $\mathbf{L}^{t}=\left[\begin{array}{c}
a\mathbf{1}_{m-1}^{t}\\
b+\mathbf{I}_{m-1}
\end{array}\right]$, with $a=\frac{1}{\sqrt{m}}$ and $b=-\left(\frac{1}{\left(m-1\right)}\left(1+\frac{1}{\sqrt{m}}\right)\right)$
is an orthonormal basis for the set of zero-mean vectors $\mathcal{A}_{m}=\left\{ \mathbf{u}\in\R^{m}|\sum u_{i}=0\right\} $. 
\end{prop}
The proof can be found in Appendix \ref{sub:An-orthogonal-basis-zero-mean}.
With this choice of $\mathbf{L}$ the matrix $\mathbf{B}$ equals

\begin{equation}
\mathbf{B}=\left[\begin{array}{ccccc}
a & a & a & a & a\\
a & b+1 & b & b & b\\
a & b & b+1 & \ldots & b\\
a & \vdots & \vdots & \ddots & b\\
a & b & b & \ldots & b+1
\end{array}\right]\label{eq:B-matrix}
\end{equation}

and it is easy to check that $\mathbf{B}^{t}\mathbf{B}=\mathbf{I}$,
$\mathbf{B}\mathbf{e}=\left[\begin{array}{cccc}
1 & 0 & \ldots & 0\end{array}\right]$, and that in addition $\mathbf{B}^{t}=\mathbf{B}$, so that $\mathbf{X}\mathbf{B}=\left(\mathbf{B}^{t}\mathbf{X}^{t}\right)^{t}=\left(\mathbf{B}\mathbf{X}\right)^{t}$.
Due to the specific structure, a single MVP with $\mathbf{B}$ has
cost $\mathcal{O}\left(m\right)$ (which is the cost of multiplying
the entries by $a$, by $b$ and computing the sum).

\subsubsection{Computing Gaussian densities and the determinant}

In applications we will need to compute Gaussian densities of the
following form: given $\mathbf{x}\sim\N\left(0,\bS\right)$ 

\begin{eqnarray}
p\left(\mathbf{x}\right) & = & \left(2\pi\right)^{-mn/2}\exp\left(-\frac{1}{2}\mathbf{x}^{t}\bS^{-1}\mathbf{x}-\log\det\bS^{-1}\right)\nonumber \\
 & = & \left(2\pi\right)^{-mn/2}\exp\left(-\frac{1}{2}\mathbf{x}^{t}\mathbf{G}^{-t}\mathbf{G}^{-1}\mathbf{x}-\log\det\bS^{-1}\right)\nonumber \\
 & = & \left(2\pi\right)^{-mn/2}\exp\left(-\frac{1}{2}\mathbf{z}^{t}\mathbf{z}-\log\det\bS^{-1}\right)\label{eq:Gaussian-density}
\end{eqnarray}

where $\mathbf{z}=\mathbf{G}^{-1}\mathbf{x}$ is obtained via whitening.
In addition, the log-determinant of $\bS$ is needed. Similar matrices
have the same determinant, so that $\mbox{det}\left(\mathbf{R}\bS\mathbf{R}^{t}\right)=\mbox{det}\left(\bS\right)$.
Further, for block-diagonal matrices, $\log\det\left(\mbox{bdiag}\mathbf{A}_{i}\right)=\sum\log\det\mathbf{A}_{i}$,
which together with Theorem 1 implies that 
\begin{equation}
\log\det\bS=\log\det\left(\mathbf{A}+m\mathbf{K}\right)+\left(m-1\right)\log\det\mathbf{A}\label{eq:log-det-sigma}
\end{equation}

These two log-determinants can be computed at $\O\left(n\right)$
cost from the Cholesky or eigendecompositions of $\mathbf{A}+m\mathbf{K}$
and $\mathbf{A}$, which are needed anyway for the computation of
the square root of $\bS$.

\subsubsection{Summary}

The cost of computing matrix square roots for quasi-Kronecker matrices
is dominated by the initial cost of computing two $\O\left(n^{3}\right)$
decompositions. The resulting factors have $\O\left(n^{2}\right)$
storage cost and no further storage is required. Further operations
(whitening, correlating, computing the determinant) come at much lower
cost ($\O\left(m^{2}n\right)$ for whitening and correlating, $\O\left(n\right)$
for the determinant).

\section{Applications\label{sec:Applications}}

We first describe how to use our techniques in the linear-Gaussian
setting. Under the assumption of Gaussian observations the marginal
likelihood of the hyperparameters can be computed efficiently, and
we show in addition how to compute the gradient of the marginal likelihood
in $\O\left(n^{3}+mn^{2}\right)$ cost. The resulting algorithm scales
very well with $m$.

When the observations are not Gaussian, the marginal likelihood is
unavailable in closed form. There are many ways to implement inference
in this case, including variational inference \citep{OpperArchambeau:VariationalGaussianApproxRev},
nested Laplace approximations \citep{Rue:INLA}, expectation propagation
\citep{Minka:EP}, various form of Markov Chain Monte Carlo \citep{Neal:MonteCarloImplGaussProcModelsBayesRegClass,RueHeld:GMRFTheoryApplications,Murray:EllSliceSampling},
or some combination of methods. We study in detail an example from
neuroscience, smoothing of repeated spike train data, where the observations
are Poisson distributed. We use a simple Laplace approximation, but
our methods can be applied in a similar way to more sophisticated
approximate inference or for MCMC sampling.

As stated in the introduction, we assume throughout that there are
$n\times m$ observations, corresponding to $m$ realisations of a
latent process observed on a grid of size $n$, with grid points $x_{1}\ldots x_{n}$.
The grid may be irregular but it is essential that it be constant
\emph{across realisations}, otherwise the covariance matrix does not
have the requisite form. The observations may be represented as a
$n\times m$ matrix $\mathbf{Y}$, or equivalently as a vector $\mathbf{y}=\mbox{vec}\left(\mathbf{Y}\right)$.

\subsection{Gaussian observations\label{sub:Gaussian-observations}}

In the Gaussian case the conditional distribution of the observations
is assumed to be iid, with $\mathbf{y}_{ij}\sim\N\left(g_{ij},\sigma^{2}\right)$.
To simplify the notation we will sometimes absorb the noise variance
$\sigma^{2}$ into the hyperparameters $\bt$.

Depending on the scenario the goal of the analysis varies, and we
might seek to estimate the latent functions\\
 $g_{1}\left(x\right),\ldots,g_{m}\left(x\right)$, the latent mean
function $f\left(x\right)$, or to make predictions for an unobserved
function $g_{m+1}\left(x\right)$, etc. How we treat the hyperparameters
may vary as well. In a ``mixed modelling'' spirit we might just
want to compute their maximum likelihood estimate, or instead we may
want to perform a full Bayesian analysis and produce samples from
the posterior over hyperparameters, $p\left(\bt|\mathbf{y}\right)$. 

In all of these cases the main quantities of interest are:
\begin{enumerate}
\item The (log-) marginal likelihood $\log\, p\left(\mathbf{y}|\bt,\sigma^{2}\right)$
and optionally its derivatives with respect to the hyperparameters
\item The conditional regressions $p\left(\mathbf{f}|\mathbf{y},\bt,\sigma^{2}\right)$
and $p\left(\mathbf{g}|\mathbf{y},\bt,\sigma^{2}\right)$
\end{enumerate}
Most of the other quantities are simply variants of the above and
we will not go explicitly through all the calculations.

\subsubsection{Computing the marginal likelihood and its derivatives\label{sub:gaussian-marginal-likelihood}}

The marginal likelihood can be computed easily by noting that $\mathbf{y}=\mathbf{g}+\epsilon$,
with $\epsilon\sim\N\left(0,\sigma^{2}\mathbf{I}\right)$, so that
the marginal distribution of $\mathbf{y}$ is 
\[
\mathbf{y}\sim\mathcal{N}\left(0,\bS+\sigma^{2}\mathbf{I}\right)
\]

and that if $\bS$ has rQK structure, so does $\bS'=\bS+\sigma^{2}\mathbf{I}$.
Specifically, $\bS'=\mbox{rQK}\left(\mathbf{A}+\sigma^{2}\mathbf{I},\mathbf{K}\right)=\mbox{rQK}\left(\mathbf{A}',\mathbf{K}\right)$,
and we may use formula (\ref{eq:Gaussian-density}) directly to compute
$\log\, p\left(\mathbf{y}|\bt,\sigma^{2}\right)$. 

For maximum likelihood estimation of the hyperparameters (and for
Hamiltonian Monte Carlo) it is useful to compute the gradient of $\log\, p\left(\mathbf{y}|\bt,\sigma^{2}\right)$
with respect to the hyperparameters. To simplify notation we momentarily
absorb the noise variance $\sigma^{2}$ into the hyperparameters $\bt$
and note $\bS$ the covariance matrix of $\mathbf{y}.$

An expression for the derivative with respect to a single hyperparameter
$\theta_{j}$ is given by \citet{RasmussenWilliamsGP} (eq. 5.9):

\begin{equation}
\frac{\partial}{\partial\theta_{j}}\log\, p\left(\mathbf{y}|\bt\right)=\frac{1}{2}\left(\mathbf{y}^{t}\bS^{-1}\left(\frac{\partial}{\partial\theta_{j}}\bS\right)\bS^{-1}\mathbf{y}-\mbox{tr}\left(\bS^{-1}\left(\frac{\partial}{\partial\theta_{j}}\bS\right)\right)\right)\label{eq:grad-marginal-lik}
\end{equation}

It is straightforward to verify that if $\bS$ is rQK, so is its derivative
$\frac{\partial}{\partial\theta_{j}}\bS$, which implies that $\bS^{-1}\left(\frac{\partial}{\partial\theta_{j}}\bS\right)$
is also a rQK matrix (see section \ref{sub:Block-rotated-form}).
Using that property along with theorem 1, some further algebra shows:

\begin{equation}
\mbox{tr}\left(\bS^{-1}\left(\frac{\partial}{\partial\theta_{j}}\bS\right)\right)=\mbox{tr}\left(\left(\mathbf{A}+m\mathbf{K}\right)^{-1}\frac{\partial}{\partial\theta_{j}}\left(\mathbf{A}+m\mathbf{K}\right)\right)+\left(m-1\right)\mbox{tr}\left(\mathbf{A}^{-1}\frac{\partial}{\partial\theta_{j}}\left(\mathbf{A}\right)\right)\label{eq:deriv-trace}
\end{equation}

Computing $\bS^{-1}\mathbf{y}$ can be done using the fast matrix
square root, and $\left(\frac{\partial}{\partial\theta_{j}}\bS\right)$
is a rQK matrix, so that the dot product $\mathbf{y}^{t}\bS^{-1}\left(\frac{\partial}{\partial\theta_{j}}\bS\right)\bS^{-1}\mathbf{y}$
in eq. \ref{eq:grad-marginal-lik} is tractable as well (with complexity
$\O\left(mn^{2}\right)$). Given that factorisations of $\left(\mathbf{A}+m\mathbf{K}\right)$
and $\mathbf{A}$ are needed to compute the marginal likelihood anyway,
the extra cost in computing derivatives is relatively small (the factorisations
have cost $\O\left(n^{3}\right)$, the rest is $\O\left(mn^{2}\right)$).

\subsubsection{Conditional regressions\label{sub:Conditional-regressions}}

The conditional regressions are given by the posterior distributions
$p\left(\mathbf{f}|\mathbf{y},\bt,\sigma^{2}\right)$ and $p\left(\mathbf{g}|\mathbf{y},\bt,\sigma^{2}\right)$.
Both distributions are Gaussian, and their mean and covariance can
be found through the usual Gaussian conditioning formulas:

\begin{eqnarray*}
E\left(\mathbf{f}|\mathbf{y},\bt,\sigma^{2}\right) & = & \mathbf{Q}_{f}^{-1}\left(\mathbf{A}+\sigma^{2}\mathbf{I}\right)^{-1}\left(\sum_{i=1}^{m}\mathbf{y}_{i}\right)\\
Cov\left(\mathbf{f}|\mathbf{y},\bt,\sigma^{2}\right) & = & \mathbf{Q}_{f}^{-1}\\
\mathbf{Q}_{f} & = & \mathbf{K}^{-1}+m\left(\mathbf{A}+\sigma^{2}\mathbf{I}\right)^{-1}
\end{eqnarray*}

which involves $\O\left(n^{3}+mn^{2}\right)$ computations, and:

\begin{eqnarray*}
E\left(\mathbf{g}|\mathbf{y},\bt,\sigma^{2}\right) & = & \sigma^{-2}\mathbf{Q}_{g}^{-1}\mathbf{y}\\
Cov\left(\mathbf{g}|\mathbf{y},\bt,\sigma^{2}\right) & = & \mathbf{Q}_{g}^{-1}\\
\mathbf{Q}_{g} & = & \bS^{-1}+\sigma^{-2}\mathbf{I}
\end{eqnarray*}

where $\mathbf{Q}_{g}$ is rQK and hence all computations are also
$\O\left(n^{3}+mn^{2}\right)$. Simultaneous confidence bands can
be obtained from the diagonal elements of $\mathbf{Q}_{f}^{-1}$ and
$\mathbf{Q}_{g}^{-1}$.

\subsubsection{Fast updates\label{sub:Fast-updates}}

``Fast'' hyperparameter updates are often available in Gaussian
process models \citep{Neal:MCMCEnsembleOfStates}, in the sense that
one may quickly recompute the marginal likelihood $p(\mathbf{y}|\bt')$
from the current value $p(\mathbf{y}|\bt)$. Usually a fast update
means being able to skip a $\O\left(n^{3}\right)$ factorisation,
and this may be done for rQK matrices as well. For example, the matrices
$\mathbf{A}$, $\mathbf{A}+\tau\mathbf{I}$ and $\mathbf{\alpha}\mathbf{A}$
(with $\alpha>\text{0}$) all have the same eigenvectors, which means
one can update certain hyperparameters without the need to recompute
the square root of $\bS$ from scratch. Although this strategy may
result in speed-ups, it requires painstaking implementation and we
have not explored it further (we note however that fast updates are
also linked to fast gradient computations, since eq. \ref{eq:grad-marginal-lik}
often simplifies). It may be worthwhile when used in combination with
Neal's \citeyearpar{Neal:MCMCEnsembleOfStates} strategy of separating
fast and slow variables.

\subsubsection{Results}

\paragraph{Specialised vs. generic factorisation methods }

We first checked that implementing our specialised matrix factorisations
is indeed worth the trouble (for reasonable problem sizes). Generic
matrix factorisation techniques have been greatly optimised over the
years and are often faster than an inefficient implementation of a
specialised routine. 

We implemented our matrix factorisation formulae in R and measured
computing times for Gaussian densities $p(\mathbf{x}|\bS)$. The non-specialised
routine needs to first form the matrix $\bS=\mbox{rQK}\left(\mathbf{A},\mathbf{K}\right)$
explicitly, and uses Cholesky factorisation to evaluate the density.
Our specialised implementation comes under two variants, one based
on Cholesky factors, and the other based on eigendecompositions, as
explained in section \ref{sub:A-fast-algorithm-square-roots}. The
algorithms were implemented in R and all benchmarks were executed
on a standard desktop PC running R 3.0.1 under Ubuntu.

Figure \ref{fig:Relative-efficiency} and \ref{fig:Relative-eff-loglog}
summarise the results. We varied $m$, the number of functions, for
different values of $n$, the grid size. General-purpose factorisation
scales with $\O\left(m^{3}n^{3}\right)$, compared to our method,
which scales with $\O\left(n^{3}+mn^{2}\right)$. What we expect,
and what Fig. X verifies, is that our method should scale much better
with $m$, as indeed it does: it is faster for all but the smallest
problem sizes (e.g. with $n=10$ and $m=4$ our method is not worth
the bother). With $n=100$ we do better even with $m$ as low as 2.
What the results also underscore is the relative inefficiency of eigenvalue
decompositions, relative to Cholesky factorisations. The eigenvalue
decomposition remains of interest for regular grids, since in this
case the Fourier decomposition can be used to factorise the covariance
matrix in $\O\left(n\log n\right)$ operations \citep{Paciorek:BayesSmoothGPUsingFourierBasisFunctions},
or if iterative methods are used in order to get an approximate factorisation
\citep{Saad:IterMethodsSparseLinearSystems}. 

\begin{figure}
\begin{centering}
\includegraphics[width=10cm]{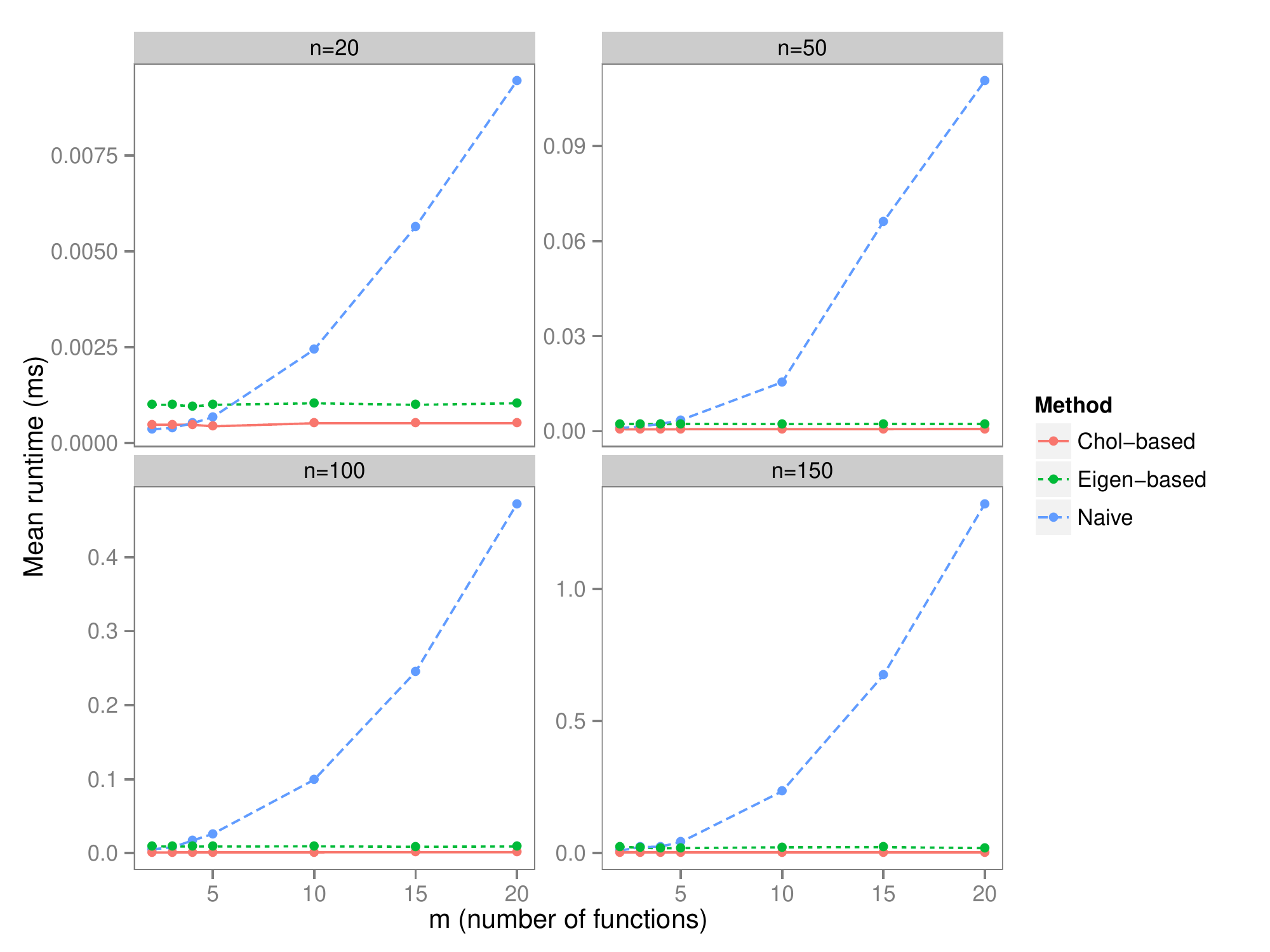}
\par\end{centering}

\caption{Relative efficiency of specialised versus generic methods for computing
values of a multivariate Gaussian density. The blue line (``Naïve'')
represents the runtime of a method based on a generic Cholesky factorisation
of the covariance matrix. The other two lines represent runtimes of
the specialised algorithms described here. One is based on an eigenvalue
decomposition, the other on a Cholesky factor. Specialised methods
are orders of magnitude faster for realistic problem sizes. \label{fig:Relative-efficiency}.
See also Fig \ref{fig:Relative-eff-loglog}.}
\end{figure}

\begin{center}
\begin{figure}
\begin{centering}
\includegraphics[width=10cm]{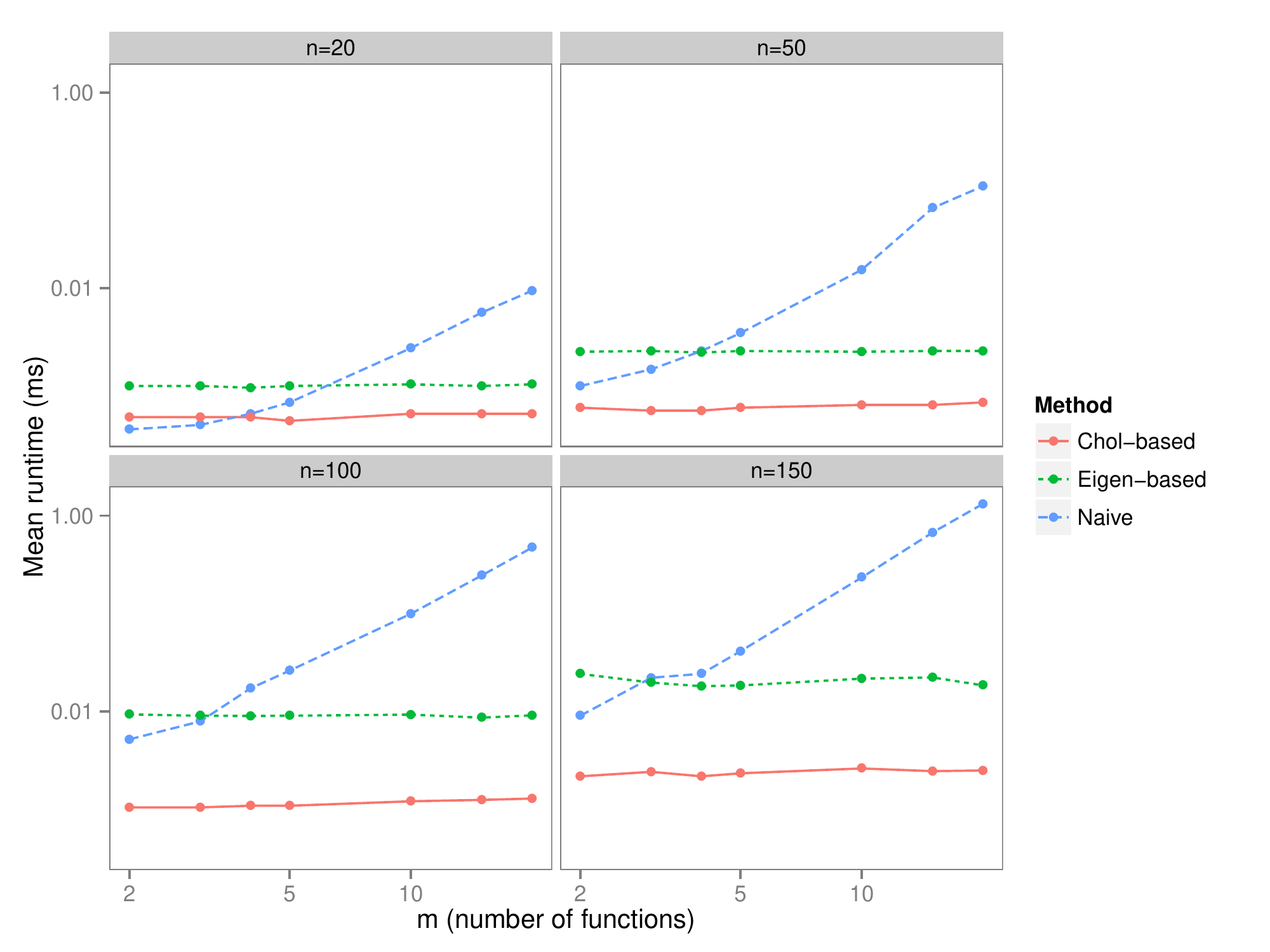}
\par\end{centering}

\caption{Same data as in fig. \ref{fig:Relative-efficiency}, replotted in
log-log coordinates. \label{fig:Relative-eff-loglog}}

\end{figure}

\par\end{center}

\paragraph{Illustration in a joint smoothing problem}

We illustrate the application of our method in a joint smoothing problem.
We generated data according to the following model:

\begin{eqnarray*}
f(x) & = & \sin\left(12x\right)+\sin\left(24x\right)\\
g_{i}\left(x\right) & = & f\left(x\right)+w_{i1}\mbox{cos}\left(6x\right)+w_{i2}\mbox{cos}\left(3x\right)\\
w_{ik} & \sim & \N\left(0,4\right)\\
y_{ij} & = & g_{i}\left(x_{j}\right)+\epsilon_{ij}\\
\epsilon_{ij} & \sim & \N\left(0,\sigma^{2}\right)
\end{eqnarray*}

meaning that the individual functions were smooth perturbations around
a fixed mean function $f(x)$, as shown on Fig. \ref{fig:Gaussian-data}. 

\begin{center}
\begin{figure}
\begin{centering}
\includegraphics[width=10cm]{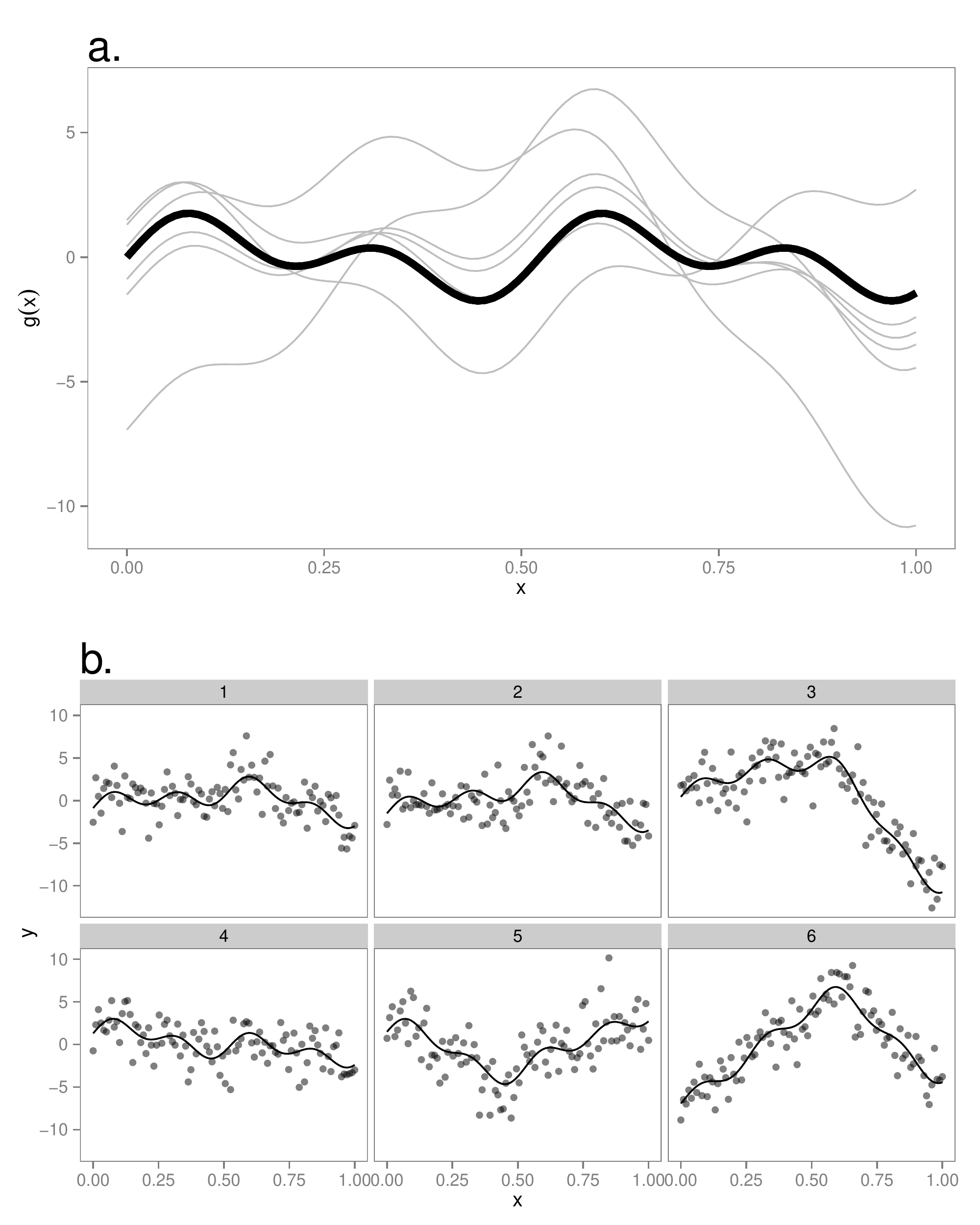}
\par\end{centering}

\caption{Data used in the simulations for the joint smoothing problem. \textbf{a}.
The mean function $f(x)$ (thick lines) with 6 variants $g_{1}\left(x\right)\ldots g_{6}\left(x\right)$
(thin lines) \textbf{b. }The data are noisy observations of the $g_{i}$'s
(dots represent datapoints, lines correspond to the latent functions
$g_{1}\left(x\right)\ldots g_{6}\left(x\right)$). In this example
$n=100$ and the grid is regular.\label{fig:Gaussian-data}}
\end{figure}

\par\end{center}

We use a regular grid of points $x_{i}\in(0,1)$.

We fitted a generic additive Gaussian process model, specifically:

\begin{eqnarray*}
f(x) & \sim & GP\left(0,\kappa_{K}\right)\\
g_{i}\left(x\right) & \sim & GP\left(0,\kappa_{A}\right)\\
y_{ij} & = & g_{i}\left(x_{j}\right)+\epsilon_{ij}
\end{eqnarray*}

The covariance functions $\kappa_{K}$ and $\kappa_{A}$ are Matern
covariance functions, with fixed smoothness parameter $\nu=\frac{5}{2}$.
The hyperparameters to estimate are the two length-scale parameters,
the two Matern variance parameters, and the noise variance $\log\sigma^{2}$.
We note $\bt$ the vector of hyperparameters. The most straightforward
estimation strategy is to maximise $\log\, p\left(\mathbf{y}\vert\bt\right)$
(maximum likelihood, ML), or alternatively $\log\, p\left(\mathbf{y}\vert\bt\right)+\log p\left(\bt\right)$
(maximum a posteriori, MAP). MCMC can naturally also be used to sample
from $p(\bt|\mathbf{y})$, and we compared both methods.

For MAP/ML we found that a quasi-Newton method (limited memory BFGS,
\citealp{LiuNocedal:OnTheLimMemBFGSLargeScaleOpt}) works very well
(we used the fast analytical gradient described above), converging
typically in \textasciitilde{} 30 iterations. The standard Gaussian
approximation of $p(\bt|\mathbf{y})$ at the mode $\bt^{\star}$ can
be computed by finite differences using the analytical gradient. For
MCMC we used standard Metropolis-Hastings, with a proposal distribution
corresponding to the Gaussian approximation of the posterior. When
the data are informative enough the posterior is well-behaved and
no further tuning is necessary, but problems can arise if e.g. the
noise is very high and there are few measurements. 

MAP inference is extremely fast (see fig. \ref{fig:Runtime-of-MAP}),
and remains tractable with extremely large datasets: with a grid size
of $n=1,000$ and $m=100$ latent functions, MAP inference takes a
little over two minutes on our machine. MCMC is an order of magnitude
slower but still rather tractable (it is likely that large speed-ups
could be obtained using a modern Hamiltonian Monte Carlo method).
Below we compare the results of MCMC and MAP inference for the data
shown on fig. X. 

Given a MAP estimate of the hyperparameters, inference for $f$ and
the $g_{i}$'s proceeds using the conditional posteriors $p(\mathbf{f}|\bt^{\star},\mathbf{y})$
and $p(\mathbf{g}|\bt^{\star},\mathbf{y})$. This tends to underestimate
the uncertainty in the latent processes, but the underestimation is
not dramatic in this example. The simultaneous confidence bands for
MAP and full MCMC inference are compared on fig. \ref{fig:Fitted-mean-function}
(appendix \ref{sub:Computing-simultaneous-confidence-bands} explains
how the confidence bands were estimated). 

\begin{center}
\begin{figure}
\begin{centering}
\includegraphics[width=14cm]{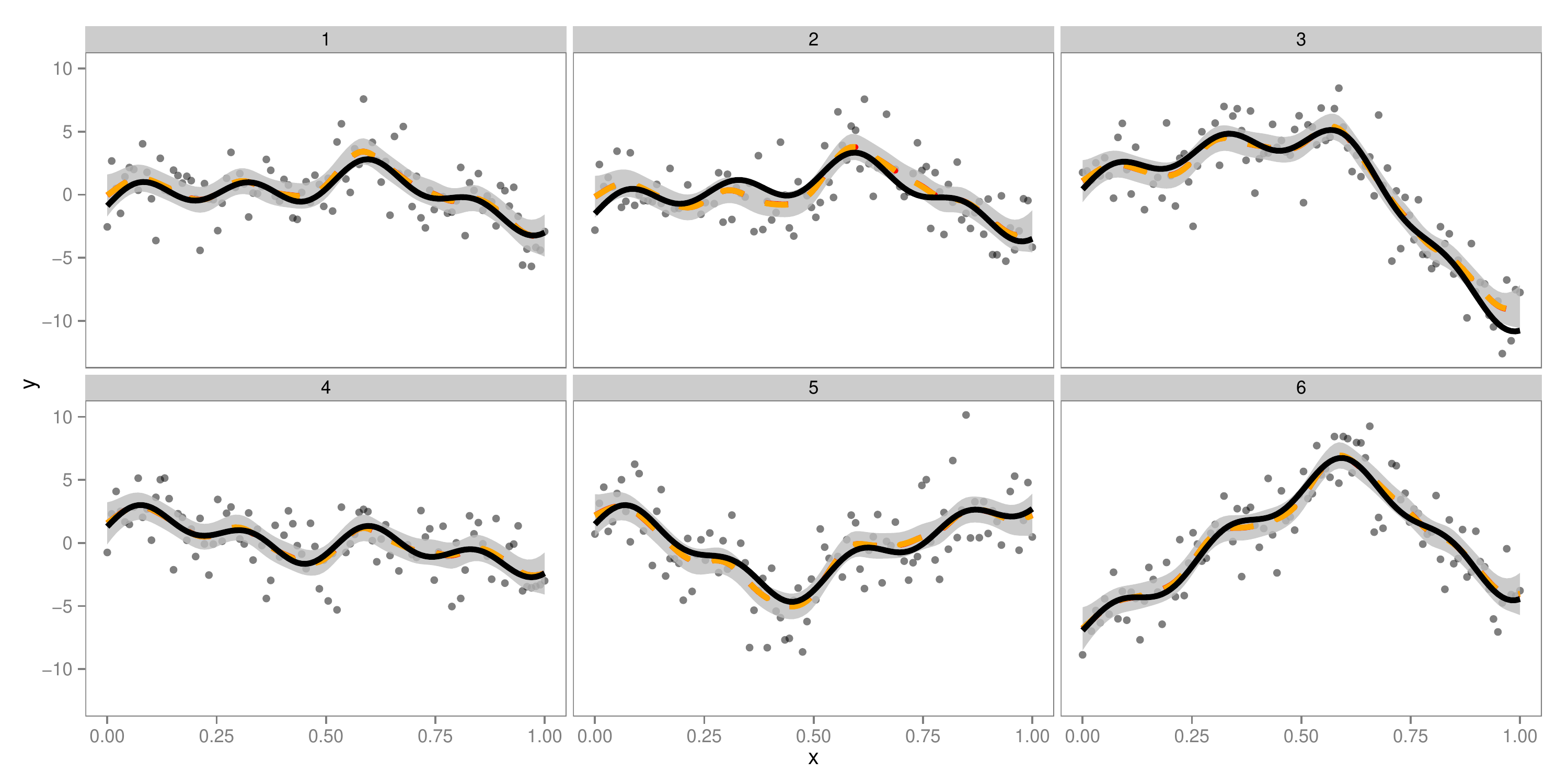}
\par\end{centering}

\caption{Fitted functions (posterior mean or MAP) vs. ground truth for the
data shown in Fig. \ref{fig:Gaussian-data}. Ground truth (actual
values of $g_{1}\left(x\right)\ldots g_{6}\left(x\right)$) is shown
as a thick black lines, MAP and posterior mean fits are shown as orange
and red dotted lines. MAP and posterior mean overlap nearly completely.
The shaded region corresponds to a pointwise posterior confidence
band from MCMC output. \label{fig:Fitted-functions-g_i}}
\end{figure}

\par\end{center}

\begin{center}
\begin{figure}
\begin{centering}
\includegraphics[width=10cm]{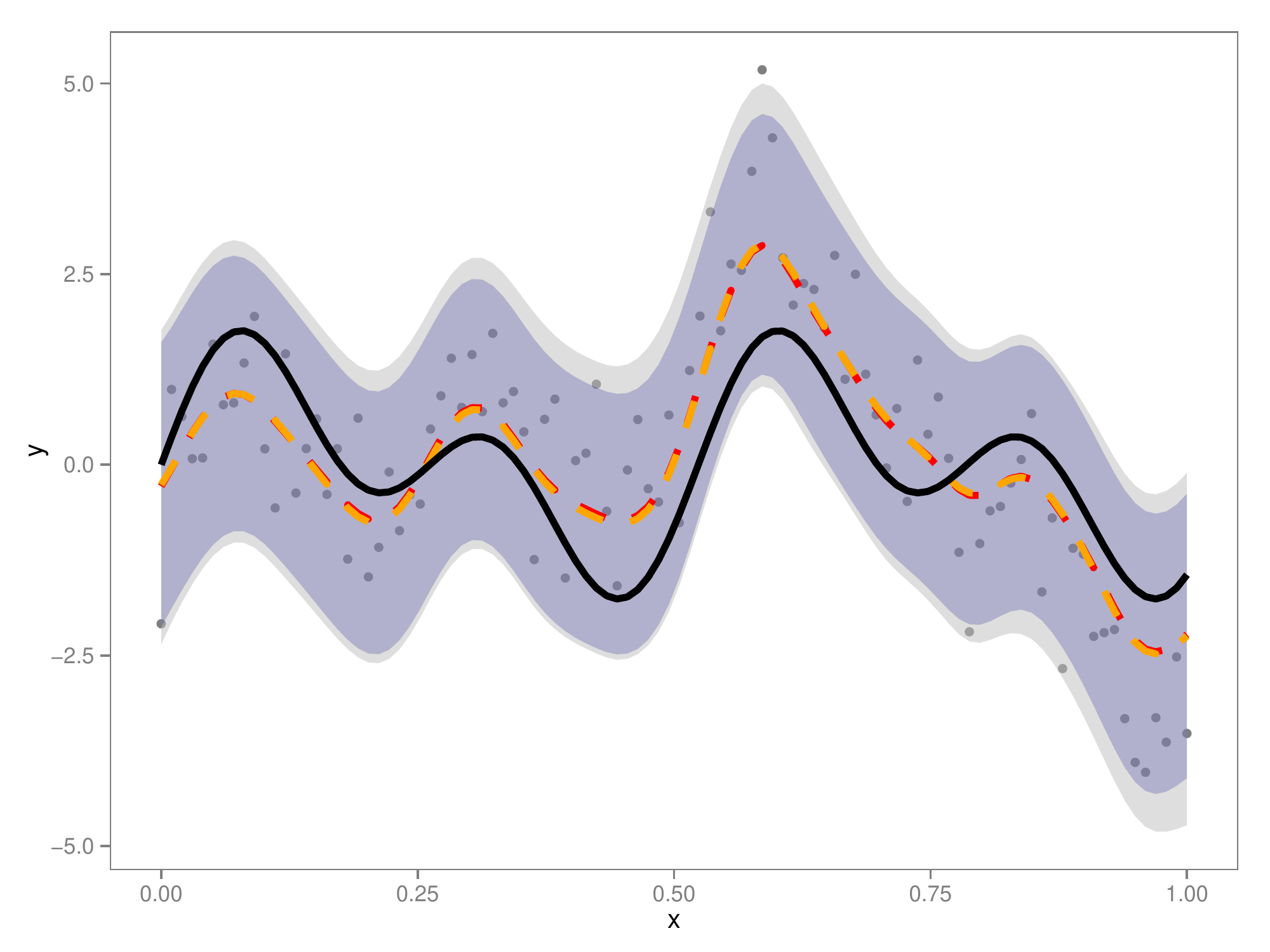}
\par\end{centering}

\caption{Posterior inference for the mean function $f\left(x\right)$ (same
data as shown in fig. \ref{fig:Gaussian-data} and \ref{fig:Fitted-functions-g_i}).
Individual dots represent mean observation values at a given point
(i.e. $\frac{1}{m}\sum_{i=1}^{m}y_{ij}$, which equals $f(x_{j})$
in expectation). The true value of $f\left(x\right)$ is shown as
a thick black line, MAP and posterior estimates are in orange and
red respectively. Two 95\% confidence bands are displayed: the lighter
one is from MCMC inference, the darker one is an approximate interval
obtained from the MAP estimate. The latter neglects posterior uncertainty
over the hyperparameters and is therefore narrower, although the inference
is not dramatic in this example, and both intervals contain the true
value. Note that under the functional additive model the mean observations
are not distributed i.i.d. around the mean function, making inference
more difficult. \label{fig:Fitted-mean-function}}

\end{figure}

\par\end{center}

\begin{figure}
\begin{centering}
\includegraphics[width=10cm]{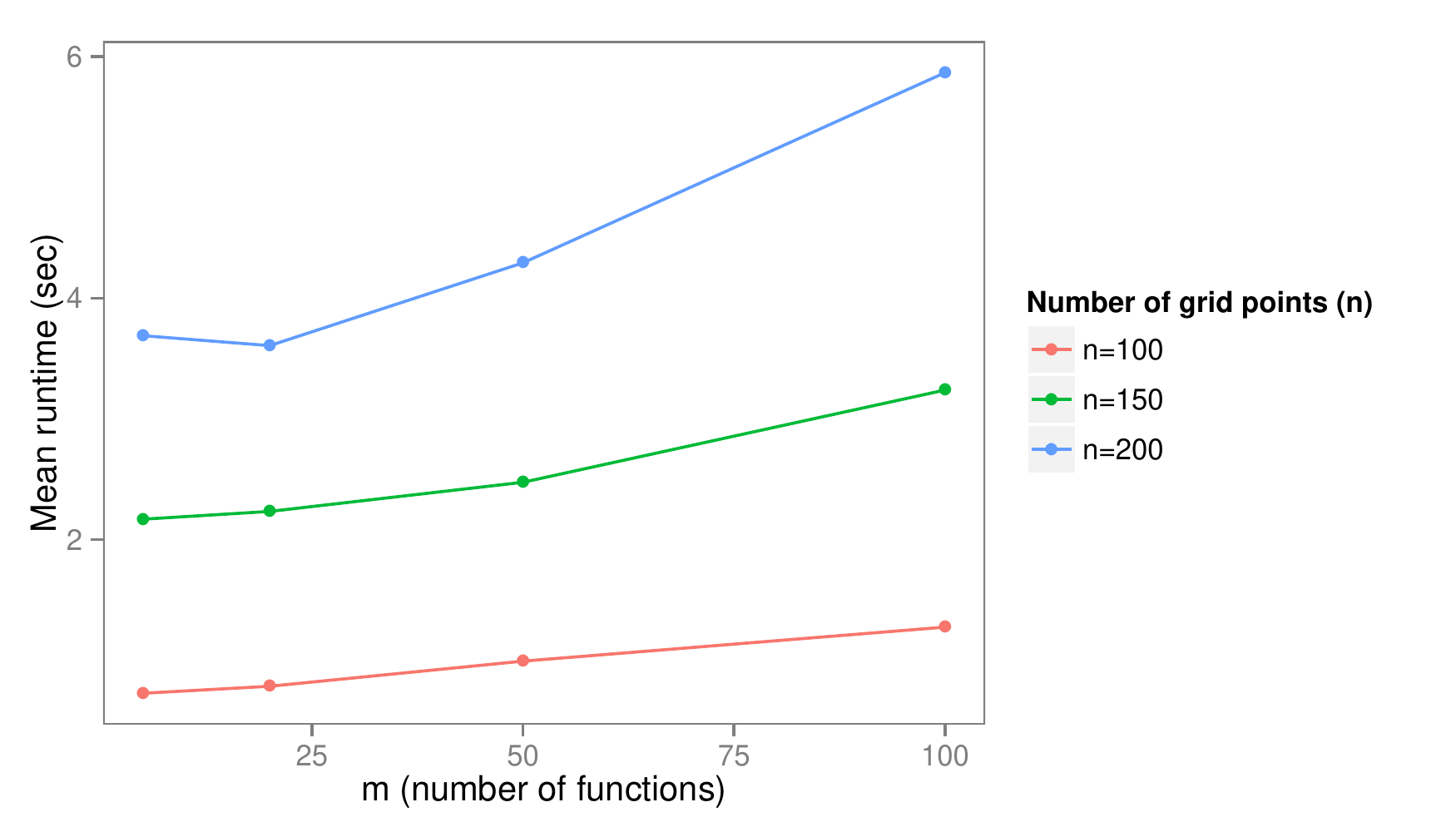}
\par\end{centering}

\caption{Runtime of MAP estimation of the hyperparameters (Gaussian case).
We simulated data with different grid sizes ($n$) and different number
of functions ($m$), and used numerical optimisation (L-BFGS-B) to
estimate the hyperparameters. We show here average runtimes obtained
over 30 repetitions with different random datasets. The method shows
excellent scaling with $m$, the largest cost factor being again the
$\O\left(n^{3}\right)$ matrix factorisations.\label{fig:Runtime-of-MAP}}
\end{figure}

\subsection{Non-Gaussian observations: application to spike count data\label{sub:Spike-train-data}}

In this section we show how to apply our technique to LGMs with non-Gaussian
likelihoods, using an example where the data are Poisson counts. The
Poisson LGM is one of the most popular kinds, especially in spatial
statistics applications \citep{Illian:StatAnalysisSpatPointPatterns,Barthelme:ModelingFixationLocationsSpatialPointProc}.
Here we focus on a application to neuroscience, specifically spike
count data. 

Neurons communicate by sending electrical signals (action potentials,
also called spikes) to one another. Each spike is a discrete event,
and the occurrence of spikes may be recorded by placing an electrode
in the vicinity of a cell. Neurons respond to stimuli by modulating
how much they spike. An excited neuron increases its spike rate, an
inhibited neuron decreases it. In experiments where spikes are recorded
a typical task is to determine how the spike rate changes over time
in response to a stimulus. 

We use data provided by \citet{PouzatChaffiol:AutomaticSpikeTrainAnalysis},
which consist in external recordings of a cell in the antennal lobe
of a locust, an area that responds to olfactory stimulation. Part
of the data is shown on Fig. \ref{fig:Spike-train-data}. The animal
was stimulated by an odorant (terpineol), causing a change in the
spike rate of the cell. Stimulation was repeated over the course of
20 successive trials. We look at spiking activity occurring between
-2 and +2 sec. relative to stimulus onset. 

The data are spike counts: we simply count the number of spikes that
occurred in a given time bin. Neurons are noisy , and statistical
models for spike trains are generally variations on the Poisson model
, which assumes that the spike count at time $t$ follows a Poisson
distribution with rate $\lambda\left(t\right)$ \citep{Paninski:StatModelsNeuralEncodingDecoding}.
The goal is to infer the underlying rate $\lambda\left(t\right)$,
which may vary spontaneously, drift from one trial to the next, and
change in response to the stimulus. 

We set up the following hierarchical model:

\begin{eqnarray*}
f\left(t\right) & \sim & GP\left(0,\kappa\right)\\
g_{i}\left(t\right) & \sim & GP\left(f(t),\kappa\right)\\
\lambda_{i}\left(t\right) & = & \exp\left(g_{i}\left(t\right)\right)\\
y_{ij} & \sim & \mbox{Poi}\left(\delta\lambda_{i}\left(t_{j}\right)\right)
\end{eqnarray*}

where $\delta$ is the time bin and $\delta\lambda_{i}\left(t_{j}\right)\approx\int_{t_{j}-\delta/2}^{t_{j+\delta/2}}\lambda_{i}\left(t\right)\mbox{d}t$
is the expected spike count in the $j$-th bin on the $i$-th trial. 

\begin{figure}
\begin{centering}
\includegraphics[width=8cm]{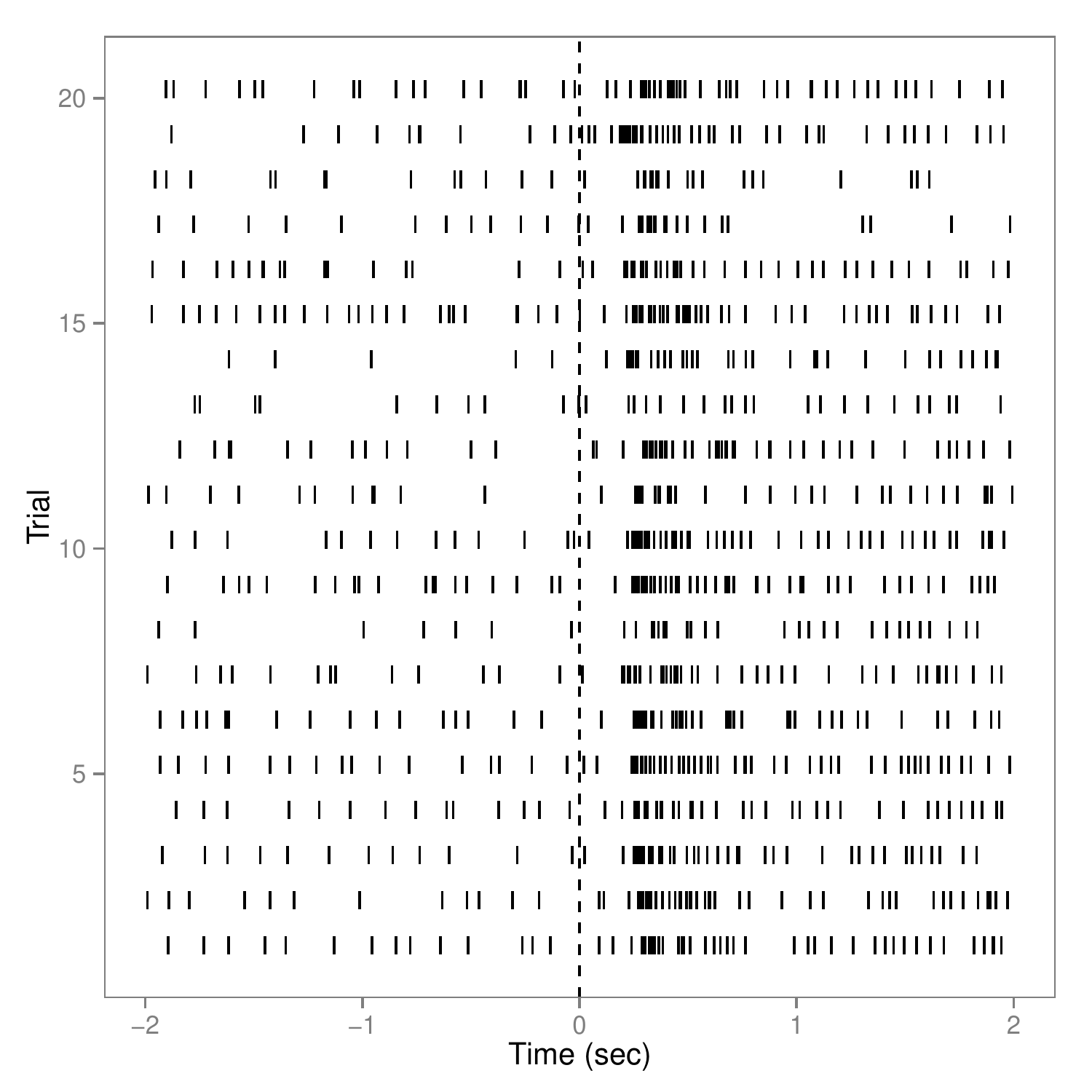}
\par\end{centering}

\caption{Spike train data used in the example: spiking responses of an antennal
lobe neuron. Each vertical bar represents a spike, with successive
trials stacked up vertically. Time $t=0$ represents stimulus onset
(stimulation with an odorant). The cell has an excitatory response:
an increase in spike rate is visible around 200ms after stimulus onset.
These data are available in the STAR package in R (\citealp{PouzatChaffiol:AutomaticSpikeTrainAnalysis},
dataset e060817terpi, cell \#1). \label{fig:Spike-train-data}}

\end{figure}

\subsection{Computing the Laplace approximation}

Contrary to the Gaussian case, for generic LGMs the posterior marginals
over the hyperparameters ($p(\mathbf{y}|\bt)$) cannot be easily computed.
The Laplace approximation usually gives sensible results \citep{Rue:INLA},
although Expectation Propagation provides a superior if more expensive
alternative \citep{NickishRasmussen:ApproxGaussianProcClass}. 

The Laplace approximation is given by:

\begin{equation}
\mathcal{L}\left(\bt\right)=f(\xs,\bt)-\frac{1}{2}\log\det\mathbf{H}\left(\xs,\bt\right)\label{eq:laplace-approx}
\end{equation}

where 
\begin{equation}
f(\mathbf{x},\bt)=\log p\left(\mathbf{y}\vert\mathbf{x}\right)+\log p\left(\mathbf{x}\vert\bt\right)\label{eq:logpost}
\end{equation}

is the unnormalised log-posterior evaluated at its conditional mode
$\xs$ , and $\mathbf{H}\left(\xs,\bt\right)$ is the Hessian of $f(\mathbf{x},\bt)$
with respect to \textbf{x }evaluated at $\xs,\bt$. 

We therefore need to (a) find the conditional mode $\xs$ for a given
value of $\bt$ and (b) compute the log-determinant of the Hessian
at the mode. The latter is possible because the Hessian turns out
to be a QK matrix, and therefore its determinant can be computed in
$\O\left(mn^{3}\right)$ using equation \ref{eq:determinant-QK}.
\begin{prop}
In LGMs the Hessian matrix of the log-posterior over the latent field
is a quasi-Kronecker matrix. \end{prop}
\begin{proof}
The second derivative of $f(\mathbf{x},\bt)$ (eq. \ref{eq:logpost})
with respect to $\mathbf{x}$ is given by:

\begin{eqnarray*}
\nabla_{\mathbf{x}}^{2}f\left(\mathbf{x},\bt\right) & = & \nabla_{\mathbf{x}}^{2}\log p\left(\mathbf{y}\vert\mathbf{x}\right)+\nabla_{\mathbf{x}}^{2}\log p\left(\mathbf{x}\vert\bt\right)\\
 & = & \mathbf{H}_{l}-\bS^{-1}
\end{eqnarray*}

where the Hessian of the log-likelihood \textbf{$\mathbf{H}_{l}$
}is diagonal (since each $y_{i}$ depends only on $x_{i}$) and $\Sigma^{-1}$
is rQK by Corrolary 2. The sum of an rQK and a diagonal matrix is
a QK matrix. 
\end{proof}
However, before we can do anything with the Hessian at the mode, we
need to find the mode. Proposition 6 implies that one may use Newton's
method \citep{NocedalWright:NumericalOptim} to do so, since $\mathbf{H}^{-1}\nabla f$
can be computed using equation \ref{eq:QK-inverse}. However, at each
Newton step we need to solve $m$ linear systems of size $n\times n$,
making it relatively expensive. We found that quasi-Newton methods
\citep{NocedalWright:NumericalOptim}, which do not use the exact
Hessian, may be more efficient provided that one chooses a smart parametrisation. 

Contrary to Newton's method, which is invariant to linear transformations
of the parameters, the performance of quasi-Newton methods depends
partly on the conditioning of the Hessian. It is therefore worthwhile
finding a transformation $\mathbf{z}=\mathbf{Mf}$ so that $\mathbf{M}\mathbf{H}\mathbf{M}^{t}$
is well-scaled. One option is to use the whitened parametrisation,
i.e. $\mathbf{M}=\bS^{-1/2}=\mathbf{G}^{-t}$ (eq. \ref{eq:whitening-transform}),
in which case 
\begin{equation}
\mathbf{M}\mathbf{H}\mathbf{M}^{t}=\mathbf{G}^{-t}\left(\bS^{-1}+\mathbf{H}_{l}\right)\mathbf{G}^{-1}=\mathbf{I}+\mathbf{G}^{-t}\mathbf{H}_{l}\mathbf{G}^{-1}\label{eq:making-prior-diagonal}
\end{equation}

Since one of the problems with Gaussian process priors is that $\bS$
can have very poor conditioning, we might hope that the above transformation
would help. We found that it does, but what is generally even more
effective is to take into account the diagonal values of $\bS^{-1/2}\mathbf{H}_{l}\bS^{-1/2}$
as well (preconditioning the problem). Further discussion of the issue
can be found in Appendix \ref{sub:Preconditioning-for-quasi-Newton}.
The Quasi-Newton method we use is limited memory BFGS in the standard
R implementation (\emph{optim}).

\subsection{Maximising the Laplace approximation}

Once we have a way of approximating $p(\mathbf{y}|\bt)$, similar
strategies apply in the general LGM case as do in the linear-Gaussian
case. We may just require an approximate MAP estimate, or we may wish
to approximately integrate out the uncertainty in $\bt$ using INLA
\citep{Rue:INLA} or exact MCMC in a pseudo-marginal sampler \citep{FilliponeGirolami:ExactApproxBayesInferenceGP}.
In any case the first step is to find the maximum of the Laplace approximation
$\mathcal{L}\left(\bt\right)$. Most optimisation methods will work,
but it helps if the gradient $\nabla\mathcal{L}\left(\bt\right)$
can be computed. It turns out that this is possible albeit rather
more expensive than in the Gaussian case, and the derivation is given
in Appendix \ref{sub:Derivative-of-the-Laplace-approx}.

\subsection{Results}

We first used the same point process models as in section \ref{sub:Gaussian-observations}
(Matern 5/2 for both $\mathbf{K}$ and $\mathbf{A}$, with four hyperparameters).
We used time bins of 2ms, giving a total of 200 time bins per spike
train and 4,000 latent variables in total. We implemented the complete
algorithm in R, and maximising the Laplace approximation takes a very
reasonable 40 sec. The fitted intensity functions $\lambda_{i}\left(t\right)=\exp\left(g_{i}\left(t\right)\right)$
are shown in fig. \ref{fig:Fitted-intensity-functions-stationary}.
There is a sharp increase in rate following stimulus onset, but little
variability across trials. Stationary Gaussian processes such as the
Matern process are not very well adapted to the estimation of functions
that jump about a lot, and we feared that the ripples visible in the
fits (e.g. prior to stimulus onset) could be artefactual. In other
words, in order not to oversmooth around the jump, the model undersmooths
in other regions. 

As an alternative, we formulated a model that allows nonstationarity
in the mean function $f(t)$, assuming:

\begin{eqnarray*}
f(t) & = & a(t)+m(t)b(t)\\
m(t) & = & \exp(-\frac{\left(t-t_{0}\right)^{2}}{s_{t}^{2}})
\end{eqnarray*}

where $a(t)$ and $b(t)$ are two independent Matern processes, and
$m(t)$ is a mask that limits the effect of $b(t)$ to the time around
stimulus onset. The net effect is that extra variability is allowed
around the time of the jump (see \citealp{Bornn:ModelNonstatProcessesDimensionExpansion}
for a more sophisticated approach to nonstationarity). We set $t_{0}=0.3$
and $s_{t}=0.2$ by eye. Adding hyperparameters for $b(t)$ brings
the total number of hyperparameters to 6, and fitting the nonstationary
model takes about 2 minutes on our machine. The results, shown on
fig. \ref{fig:Fits-nonstationary}, indicate that pre-onset ripples
are indeed most likely artefactual. 

\begin{figure}

\begin{centering}
\includegraphics[width=10cm]{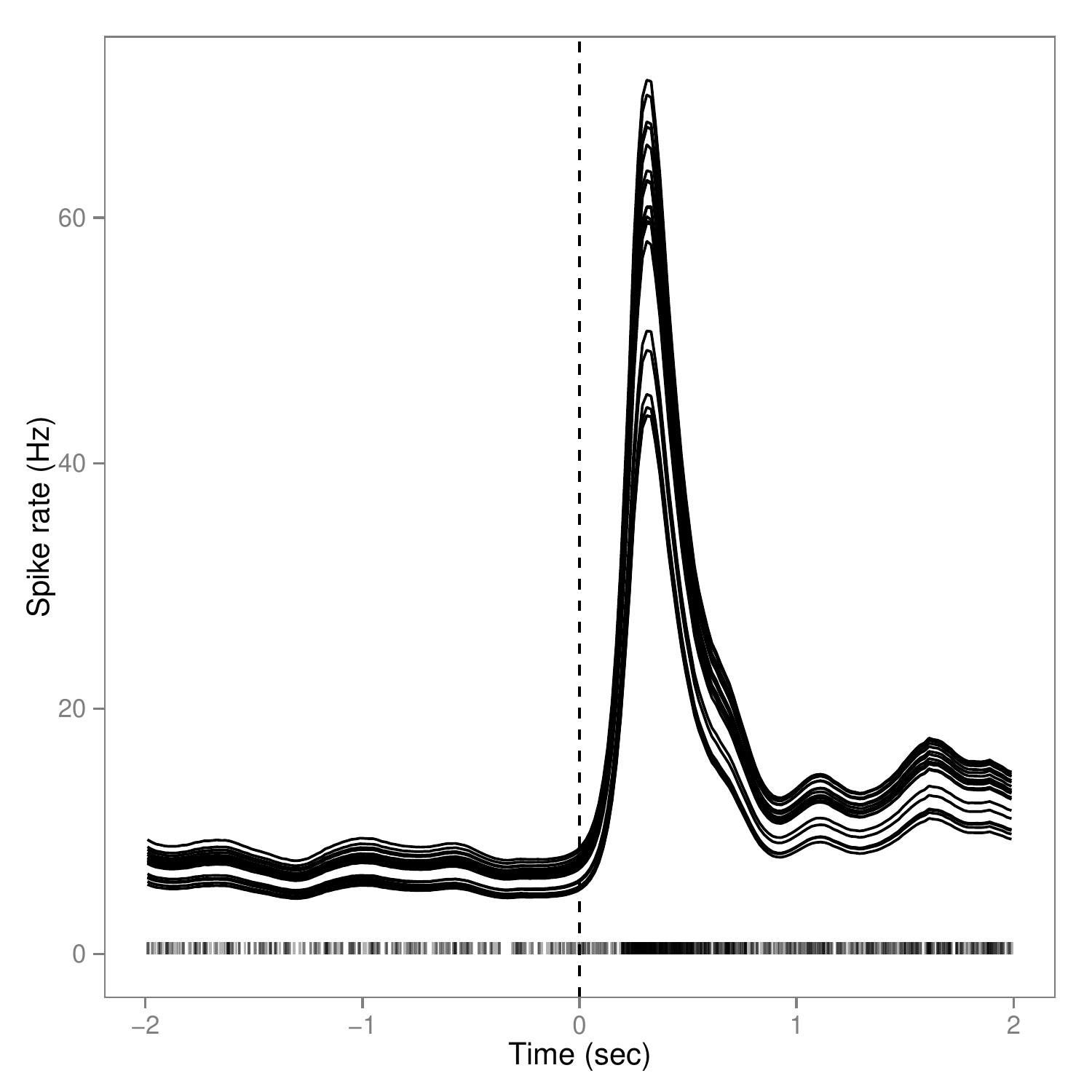}
\par\end{centering}

\caption{Fitted intensity functions for the spike train data of fig. \ref{fig:Spike-train-data},
using stationary kernels. Each line represents the fitted function
for a given trial. The strip plot at the bottom shows every spike
in the dataset (over all trials). The ``ripples'' visible before
stimulus onset and 1 sec. after onset are probably artefacts of the
stationarity assumption, see next figure.\label{fig:Fitted-intensity-functions-stationary}}

\end{figure}

\begin{figure}
\begin{centering}
\includegraphics[width=10cm]{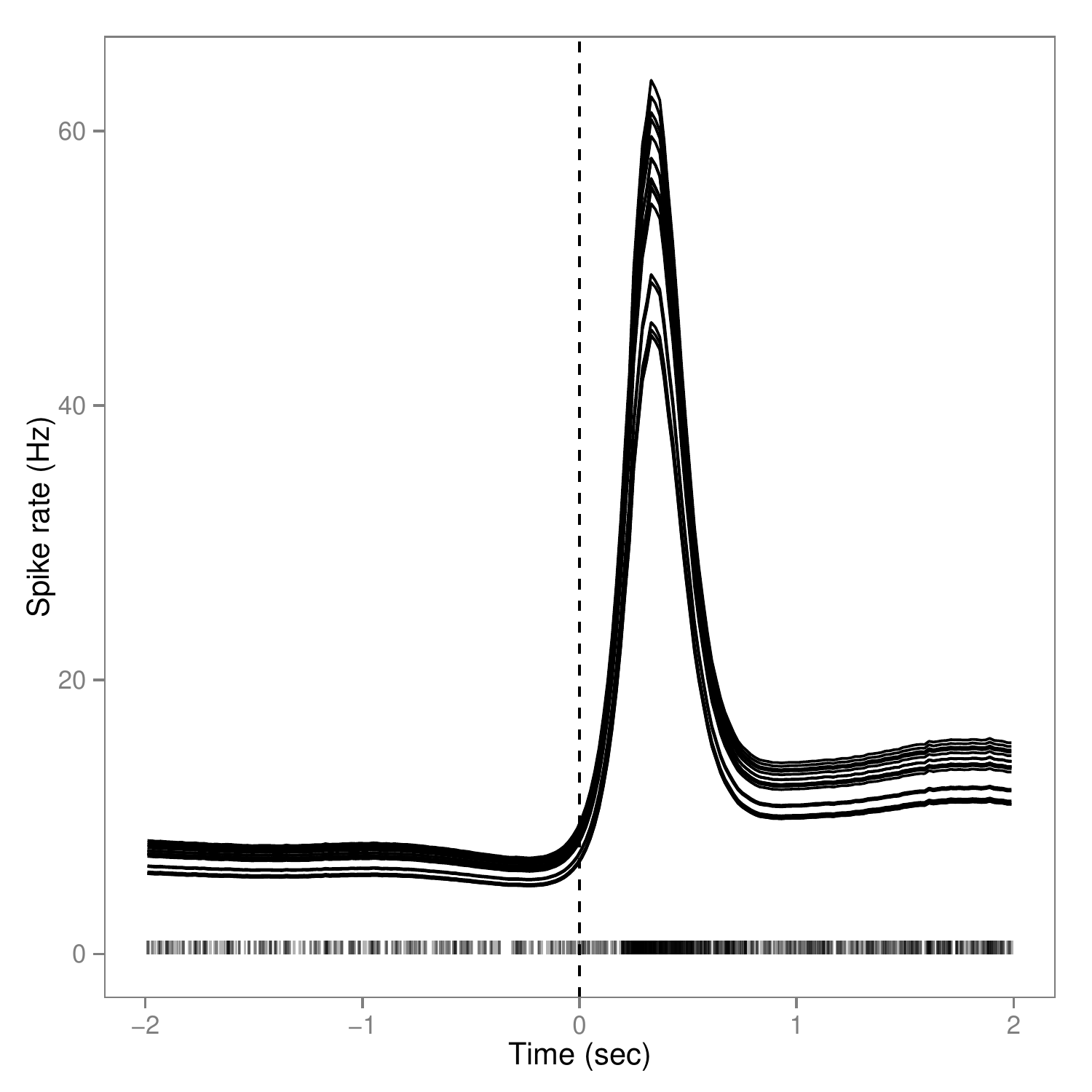}
\par\end{centering}

\caption{Same data as in fig. \ref{fig:Fitted-intensity-functions-stationary},
fitted with a non-stationary kernel (see text). The ripples visible
in fig. \ref{fig:Fitted-intensity-functions-stationary} have disappeared.\label{fig:Fits-nonstationary}}

\end{figure}

\section{Conclusion}

We have shown how restricted quasi-Kronecker matrices can be block
rotated and factorised, and how this enables efficient inference in
the two-level functional additive model. A similar closed-form factorisation
for \emph{generic }quasi-Kronecker matrices eludes us despite repeated
attempts. That would certainly make an interesting topic for future
work although perhaps not our own. 

Although the algorithms we describe here scale very well with the
number of functions in the dataset, they do not scale very well with
grid size - the dreaded $\O\left(n^{3}\right)$ remains. For regular
grids circulant embeddings (Fourier) approaches can be used \citep{Paciorek:BayesSmoothGPUsingFourierBasisFunctions},
although unfortunately the Laplace approximation becomes hard to compute
in that setting. Iterative solvers \citep{Saad:IterMethodsSparseLinearSystems,Stein:StochApproxScoreFunctionsGPs}
could overcome the problem, but we leave that for future research. 

Quasi-Kronecker and restricted quasi-Kronecker matrices display a
number of appealing properties, and should probably join their block-diagonal,
circulant, and Toeplitz peers among the set of computation-friendly
positive definite matrices.

\section{Appendix }

\subsection{An orthogonal basis for zero-mean vectors\label{sub:An-orthogonal-basis-zero-mean}}

We prove below proposition 5, which we restate:
\begin{prop*}
The matrix $\mbox{\textbf{N}}=\left[\begin{array}{c}
a\mathbf{1}_{m-1}^{t}\\
b+\mathbf{I}_{m-1}
\end{array}\right]$, with $a=\frac{1}{\sqrt{m}}$ and $b=-\left(\frac{1}{\left(m-1\right)}\left(1+\frac{1}{\sqrt{m}}\right)\right)$
is an orthonormal basis for the set of zero-mean vectors $\mathcal{A}_{m}=\left\{ \mathbf{u}\in\R^{m}|\sum u_{i}=0\right\} $.\end{prop*}
\begin{proof}
For $\mathbf{N}_{m\times\left(m-1\right)}$ to be an orthogonal basis
for $\mathcal{A}$, all that is required is that $\sum_{i=1}^{m}\mbox{\textbf{N}}_{ij}=0$
for all $j$ and that $\mathbf{N}^{t}\mathbf{N}=\mathbf{I}$. To lighten
the notation, define $r=m-1$. The first condition requires 
\begin{eqnarray}
a+rb+1 & = & 0\nonumber \\
a & = & -(1+rb)\label{eq:condition-ortho-basis-1}
\end{eqnarray}
and the second condition
\begin{eqnarray*}
\mathbf{N}^{t}\mathbf{N} & = & \mathbf{I}\\
a^{2}+rb^{2}+2b+\mathbf{I} & = & \mathbf{I}
\end{eqnarray*}

implying that 
\begin{equation}
a^{2}+rb^{2}+2b=0\label{eq:condition-ortho-basis2}
\end{equation}

injecting \ref{eq:condition-ortho-basis-1} into \ref{eq:condition-ortho-basis2},
we get:

\begin{eqnarray*}
(1+rb)^{2}+rb^{2}+2b & = & 0\\
\left(r+r^{2}\right)b^{2}+\left(2+2r\right)b+1 & = & 0
\end{eqnarray*}

This is a 2nd order polynomial in $b$, and it has real roots if:

\begin{eqnarray*}
\left(2+2r\right)^{2}-4(r+r^{2}) & > & 0\\
4+4r & > & 0
\end{eqnarray*}

which is true since $r=m-1$ is non-negative.

Solving the quadratic equations yields $b=-\left(\frac{1}{r}+\frac{1}{r\sqrt{m}}\right)$
and substituting into \ref{eq:condition-ortho-basis-1} yields $a=\frac{1}{\sqrt{m}}$.

We can therefore always find values $a,b$ such that $\mathbf{N}=\left[\begin{array}{c}
a\mathbf{1}_{m-1}^{t}\\
b+\mathbf{I}_{m-1}
\end{array}\right]$ is an orthogonal basis for the linear subset of zero-mean vectors.
Moreover, we can show that the columns of \textbf{N }have unit norm:

\begin{eqnarray*}
\sum_{i}\mathbf{N}_{ij}^{2} & = & a^{2}+\left(b+1\right)^{2}+\left(m-2\right)b^{2}\\
 & = & a^{2}+\left(m-1\right)b^{2}+2b+1\\
 & = & 1
\end{eqnarray*}

where the last line is from \ref{eq:condition-ortho-basis2}. The
result implies that we can compute an orthonormal basis $\mathbf{N}$
for $\mathcal{A}$ such that matrix-vector products with $\mathbf{N}$
are $\O\left(m\right)$. 
\end{proof}

\subsection{Derivative of the Laplace approximation\label{sub:Derivative-of-the-Laplace-approx}}

Using the implicit function theorem an analytical expression for the
gradient of the Laplace approximation can be found (see \citealp{RasmussenWilliamsGP},
for a similar derivation). The Laplace approximation is given by:

\begin{equation}
\mathcal{L}\left(\bt\right)=f(\xs,\bt)-\frac{1}{2}\log\det\mathbf{H}\left(\xs,\bt\right)\label{eq:laplace-appendix}
\end{equation}

where 
\begin{equation}
f(\mathbf{x},\bt)=\log p\left(\mathbf{y}\vert\mathbf{x}\right)+\log p\left(\mathbf{x}\vert\bt\right)\label{eq:logpost-appendix}
\end{equation}

is the unnormalised log-posterior evaluated at its conditional mode
$\xs$, and $\mathbf{H}\left(\xs,\bt\right)$ is the Hessian of $f(\mathbf{x},\bt)$
with respect to \textbf{x }evaluated at $\xs,\bt$. Since $\mathbf{x}^{\star}$
is a maximum it satisfies the gradient equation:

\begin{eqnarray*}
F\left(\mathbf{x},\bt\right) & = & 0\\
F(\mathbf{x},\bt) & = & \nabla\log p\left(\mathbf{y}\vert\mathbf{x}\right)+\nabla\log p\left(\mathbf{x}\vert\bt\right)
\end{eqnarray*}

Assuming that $f(\mathbf{x},\bt)$ is twice-differentiable and concave
in $\mathbf{x}$, we can define an implicit function $\xs\left(\bt\right)$
such that $F(\xs\left(\bt\right),\bt)=0$ for all $\bt$, with derivative
\begin{eqnarray}
\frac{\partial}{\partial\theta_{j}}\xs & = & -\left(\frac{\partial}{\partial\mathbf{x}}F(\mathbf{x},\bt)\right)^{-1}\left(\frac{\partial}{\partial\bt}F(\mathbf{x},\bt)\right)\nonumber \\
 & = & -\left(\mathbf{H}\left(\xs\left(\bt\right),\bt\right)\right)^{-1}\left(\bS^{-1}\left(\frac{\partial}{\partial\theta_{j}}\bS\right)\bS^{-1}\mathbf{x}\left(\bt\right)\right)\label{eq:dxstar}
\end{eqnarray}

To simplify the notation we will note $\mathbf{G}$ the matrix $\bS^{-1}\left(\frac{\partial}{\partial\theta_{j}}\bS\right)\bS^{-1}$.
Note that the same $\mathbf{G}$ matrix appears in the gradient of
the Gaussian likelihood with respect to the hyperparameters (eq. \ref{eq:grad-marginal-lik}).
To compute the gradient of $\mathcal{L}\left(\bt\right)$, we need
the derivatives of $f\left(\xs\left(\bt\right),\bt\right)$ with respect
to$\bt$:

\[
\frac{\partial}{\partial\bt}f=\left(\frac{\partial}{\partial\xs}f\left(\xs,\bt\right)\right)\left(\frac{\partial}{\partial\theta_{j}}\xs\right)+\frac{\partial}{\partial\theta_{j}}f\left(\xs\left(\bt\right),\bt\right)
\]

where the first part is 0 since the gradient of $f$ is 0 at $\xs\left(\bt\right)$,
and 
\[
\frac{\partial}{\partial\theta_{j}}f\left(\mathbf{x},\bt\right)=\frac{\partial}{\partial\theta_{j}}\left(\log p\left(\mathbf{y}\vert\mathbf{x}\right)+\log p\left(\mathbf{x}\vert\bt\right)\right)=\frac{\partial}{\partial\theta_{j}}\log p\left(\mathbf{x}\vert\bt\right)
\]

This is the gradient of a log-Gaussian density, and we can therefore
reuse formula \ref{eq:grad-marginal-lik} for that part. 

The gradient of $\log\det\mathbf{H}\left(\xs,\bt\right)$ is also
needed and is sadly more troublesome. A formula for the derivative
of the log-det function is given in \citet{PetersenPedersen:MatrixCookbook}:

\begin{equation}
\frac{\partial}{\partial\theta_{j}}\log\det\mbox{\textbf{H}}=\mbox{Tr}\left(\mathbf{H}^{-1}\frac{\partial}{\partial\theta_{j}}\mathbf{H}\right)\label{eq:gradient-logdet}
\end{equation}

The Hessian at the mode is the sum of the Hessian of $\log p\left(\mathbf{y}\vert\xs\right)$
(which depends on $\bt$ through the implicit dependency of $\xs$
on $\bt$), and the Hessian of $\log p\left(\mathbf{x}|\bt\right)$,
which equals the inverse covariance matrix $\bS^{-1}$. Some additional
algebra yields: 
\begin{eqnarray}
\frac{\partial}{\partial\theta_{j}}\mathbf{H} & = & \left(\mbox{diag}\left(\left(\frac{\partial}{\partial\theta_{j}}\mathbf{x}_{i}^{\star}\right)\frac{\partial^{3}}{\partial^{3}x_{i}}\log p\left(\mathbf{y}\vert\mathbf{x}\right)\right)-\bS^{-1}\left(\frac{\partial}{\partial\theta_{j}}\bS\right)\bS^{-1}\right)\label{eq:gradient-hessian}\\
 & = & \left(\mathbf{D}-\mathbf{G}\right)
\end{eqnarray}

where we have assumed that the Hessian of the log-likelihood is diagonal.
Accordingly $\mathbf{H}^{-1}$ is quasi-Kronecker, and so is $\frac{\partial}{\partial\theta_{j}}\mathbf{H}$
(since it equals the sum of a rQK matrix $-\mathbf{G}$ and a diagonal
perturbation $\mathbf{D}$). Computing the trace in eq. (\ref{eq:gradient-logdet})
is therefore tractable if slightly painful, and can be done by plugging
in eq. \ref{eq:QK-inverse} and summing the diagonal elements.

\subsection{Computing simultaneous confidence bands\label{sub:Computing-simultaneous-confidence-bands}}

In this section we outline a simple method for obtaining a Rao-Blackwellised
confidence band from posterior samples of $p(\bt|\mathbf{y})$. A
simultaneous confidence band around a latent function $f(x)$ is a
vector function $c_{\alpha}\left(x\right)=\left(\begin{array}{c}
h_{\alpha}\left(x\right)\\
l_{\alpha}\left(x\right)
\end{array}\right)$, such that, for all $x$:

\begin{eqnarray*}
p\left(f\left(x\right)>h_{\alpha}\left(x\right)|\mathbf{y}\right) & < & 1-\alpha\\
p\left(f\left(x\right)\leq l_{\alpha}\left(x\right)|\mathbf{y}\right) & < & \alpha
\end{eqnarray*}

In latent Gaussian models the posterior over latent functions is a
mixture over conditional posteriors: 
\begin{equation}
p\left(f\left(x\right)|\mathbf{y}\right)=\int p(f|\mathbf{y},\bt)p(\bt|\mathbf{y})\mbox{d}\bt\label{eq:conditional-post-appendix}
\end{equation}

The conditional posteriors $p\left(f\left(t\right)|\mathbf{y},\bt\right)$
are either Gaussian or approximated by a Gaussian. If we have obtained
samples from $p(\bt|\mathbf{y})$ (as in section \ref{sub:Gaussian-observations})
we can approximate (\ref{eq:conditional-post-appendix}) using a mixture
of Gaussians, which has an analytic cumulative density function. The
c.d.f. can be inverted numerically to obtain values $h_{\alpha}\left(t\right)$
and $l_{\alpha}\left(t\right)$.

\subsection{Preconditioning for quasi-Newton optimisation\label{sub:Preconditioning-for-quasi-Newton}}

In the whitened parametrisation the Hessian matrix has the following
form (eq \ref{eq:making-prior-diagonal})

\[
\mathbf{H}_{z}\left(\mathbf{z}\right)=\mathbf{I}+\mathbf{G}^{-t}\mathbf{H}_{l}\left(\mathbf{G}^{t}\mathbf{z}\right)\mathbf{G}^{-1}
\]

Although the whitened parametrisation gives ideal conditioning for
the prior term, the conditioning with respect to the likelihood term
($\mathbf{H}_{l}$) may be degraded. Preconditioning the problem can
help tremendously: given a guess $\mathbf{z}_{0}$, one precomputes
the diagonal values of $\mathbf{d}_{0}=\mbox{diag}\left(\mathbf{I}+\mathbf{G}^{-t}\mathbf{H}_{l}\left(\mathbf{G}^{t}\mathbf{z}_{0}\right)\mathbf{G}^{-t}\right)$
and uses $\mbox{diag}\left(\mathbf{d}_{0}\right)$ as a preconditioner.
To compute $\mathbf{d}_{0}$ the following identity is useful:

\[
\mbox{diag}\left(\mathbf{U}\mbox{diag}\left(\mathbf{v}\right)\mathbf{U}^{t}\right)=\left(\mathbf{U}\odot\mathbf{U}\right)\mathbf{v}
\]

where $\odot$ is the element-wise product. In LGMs $\mathbf{H}_{l}$
is diagonal and so:

\begin{equation}
\mathbf{d}_{0}=1+\left(\mathbf{G}^{-t}\odot\mathbf{G}^{-t}\right)\mathbf{G}^{t}\mathbf{z}_{0}\label{eq:precond-diag}
\end{equation}

We inject \ref{eq:whitening-transform} into the above and some rather
tedious algebra shows that $\mathbf{d}_{0}$ can be computed in the
following way. Form the matrix $\mathbf{X}_{0}$ by stacking successive
subsets of length $n$ from $\mathbf{G}^{t}\mathbf{z}_{0}$. Rotate
using the matrix $\mathbf{B}$ (section \ref{sub:Fast-rotations})
to form another matrix: 
\[
\mathbf{Y}=\mathbf{X}_{0}\mathbf{B}
\]

then compute $\mathbf{d}_{0}$ from:

\[
\mathbf{d}_{0}=\mbox{vec}\left(\left[\begin{array}{cc}
\left(\mathbf{U}\odot\mathbf{U}\right)\mathbf{Y}_{:,1} & \left(\mathbf{V}\odot\mathbf{V}\right)\mathbf{Y}_{:,2:m}\end{array}\right]\right)
\]

where we have used Matlab notation subsets of columns of $\mathbf{Y}$,
and $\mathbf{U}$ and $\mathbf{V}$ are as in Corrolary 4.

\subsection*{Acknowledgements}

The author wishes to thank Reinhard Furrer for pointing him to \citet{HeersinkFurrer:MoorePenroseInversesQuasiKronecker},
and Alexandre Pouget for support. 

\bibliographystyle{apalike}
\bibliography{ref}

\begin{thebibliography}{}

\bibitem[Barthelm\'{e} et~al.,
  2013]{Barthelme:ModelingFixationLocationsSpatialPointProc}
Barthelm\'{e}, S., Trukenbrod, H., Engbert, R., and Wichmann, F. (2013).
\newblock Modeling fixation locations using spatial point processes.
\newblock {\em Journal of vision}, 13(12).

\bibitem[Behseta et~al., 2005]{Behseta:HierarchModelsVariabFunctions}
Behseta, S., Kass, R.~E., and Wallstrom, G.~L. (2005).
\newblock {Hierarchical models for assessing variability among functions}.
\newblock {\em Biometrika}, 92(2):419--434.

\bibitem[Bornn et~al., 2012]{Bornn:ModelNonstatProcessesDimensionExpansion}
Bornn, L., Shaddick, G., and Zidek, J.~V. (2012).
\newblock {Modeling Nonstationary Processes Through Dimension Expansion}.
\newblock {\em Journal of the American Statistical Association},
  107(497):281--289.

\bibitem[Cheng et~al., 2013]{Cheng:BayesianRegistFunctionsCurves}
Cheng, W., Dryden, I.~L., and Huang, X. (2013).
\newblock {Bayesian registration of functions and curves}.

\bibitem[Filippone and Girolami,
  2013]{FilliponeGirolami:ExactApproxBayesInferenceGP}
Filippone, M. and Girolami, M. (2013).
\newblock {Exact-Approximate} bayesian inference for gaussian processes.

\bibitem[Golub and Van~Loan, 1996]{GolubVanLoan:MatrixComputations}
Golub, G.~H. and Van~Loan, C.~F. (1996).
\newblock {\em Matrix Computations (3rd Edition)}.
\newblock Johns Hopkins University Press, 3rd edition.

\bibitem[Heersink and Furrer,
  2011]{HeersinkFurrer:MoorePenroseInversesQuasiKronecker}
Heersink, D.~K. and Furrer, R. (2011).
\newblock On moore-penrose inverses of quasi-kronecker structured matrices.
\newblock {\em Linear Algebra and its Applications}.

\bibitem[Illian et~al., 2008]{Illian:StatAnalysisSpatPointPatterns}
Illian, J., Penttinen, A., Stoyan, H., and Stoyan, D. (2008).
\newblock {\em Statistical Analysis and Modelling of Spatial Point Patterns
  (Statistics in Practice)}.
\newblock Wiley-Interscience, 1 edition.

\bibitem[Kaufman and Sain, 2009]{KaufmanSain:BayesianFuncANOVA}
Kaufman, C. and Sain, S. (2009).
\newblock {Bayesian functional ANOVA modeling using Gaussian process prior
  distributions}.
\newblock {\em Bayesian Analysis}, 5:123--150.

\bibitem[Kneip and Ramsay, 2008]{KneipRamsay:CombiningRegistrationFitting}
Kneip, A. and Ramsay, J.~O. (2008).
\newblock {Combining Registration and Fitting for Functional Models}.
\newblock {\em Journal of the American Statistical Association},
  103(483):1155--1165.

\bibitem[Liu and Nocedal, 1989]{LiuNocedal:OnTheLimMemBFGSLargeScaleOpt}
Liu, D. and Nocedal, J. (1989).
\newblock On the limited memory {BFGS} method for large scale optimization.
\newblock {\em Mathematical Programming}, 45(1-3):503--528.

\bibitem[Minka, 2001]{Minka:EP}
Minka, T.~P. (2001).
\newblock Expectation propagation for approximate bayesian inference.
\newblock In {\em UAI '01: Proceedings of the 17th Conference in Uncertainty in
  Artificial Intelligence}, pages 362--369, San Francisco, CA, USA. Morgan
  Kaufmann Publishers Inc.

\bibitem[Murray et~al., 2010]{Murray:EllSliceSampling}
Murray, I., Adams, R.~P., and MacKay, D.~J. (2010).
\newblock Elliptical slice sampling.
\newblock {\em JMLR: W\&CP}, 9:541--548.

\bibitem[Neal, 1997]{Neal:MonteCarloImplGaussProcModelsBayesRegClass}
Neal, R.~M. (1997).
\newblock {Monte Carlo Implementation of Gaussian Process Models for Bayesian
  Regression and Classification}.

\bibitem[Neal, 2011]{Neal:MCMCEnsembleOfStates}
Neal, R.~M. (2011).
\newblock {MCMC} using ensembles of states for problems with fast and slow
  variables such as gaussian process regression.

\bibitem[Nickisch and Rasmussen,
  2008]{NickishRasmussen:ApproxGaussianProcClass}
Nickisch, H. and Rasmussen, C.~E. (2008).
\newblock {Approximations for Binary Gaussian Process Classification}.
\newblock {\em Journal of Machine Learning Research}, 9:2035--2078.

\bibitem[Nocedal and Wright, 2006]{NocedalWright:NumericalOptim}
Nocedal, J. and Wright, S. (2006).
\newblock {\em Numerical Optimization (Springer Series in Operations Research
  and Financial Engineering)}.
\newblock Springer, 2nd edition.

\bibitem[Opper and Archambeau,
  2008]{OpperArchambeau:VariationalGaussianApproxRev}
Opper, M. and Archambeau, C. (2008).
\newblock The variational gaussian approximation revisited.
\newblock {\em Neural Computation}, 21(3):786--792.

\bibitem[Paciorek, 2007]{Paciorek:BayesSmoothGPUsingFourierBasisFunctions}
Paciorek, C.~J. (2007).
\newblock {Bayesian Smoothing with Gaussian Processes Using Fourier Basis
  Functions in the spectralGP Package.}
\newblock {\em Journal of statistical software}, 19(2).

\bibitem[Paninski et~al., 2007]{Paninski:StatModelsNeuralEncodingDecoding}
Paninski, L., Pillow, J., and Lewi, J. (2007).
\newblock Statistical models for neural encoding, decoding, and optimal
  stimulus design.
\newblock {\em Progress in brain research}, 165:493--507.

\bibitem[Petersen and Pedersen, 2012]{PetersenPedersen:MatrixCookbook}
Petersen, K.~B. and Pedersen, M.~S. (2012).
\newblock The matrix cookbook.
\newblock Version 20121115.

\bibitem[Pouzat and Chaffiol, 2009]{PouzatChaffiol:AutomaticSpikeTrainAnalysis}
Pouzat, C. and Chaffiol, A. (2009).
\newblock Automatic spike train analysis and report generation. an
  implementation with r, {R2HTML} and {STAR}.
\newblock {\em Journal of Neuroscience Methods}, 181(1):119--144.

\bibitem[Ramsay and Silverman, 2005]{RamsaySilverman:FunctionalDataAnalysis}
Ramsay, J. and Silverman, B.~W. (2005).
\newblock {\em Functional Data Analysis (Springer Series in Statistics)}.
\newblock Springer, 2nd edition.

\bibitem[Rasmussen and Williams, 2005]{RasmussenWilliamsGP}
Rasmussen, C.~E. and Williams, C. K.~I. (2005).
\newblock {\em {Gaussian Processes for Machine Learning (Adaptive Computation
  and Machine Learning series)}}.
\newblock The MIT Press.

\bibitem[Rue and Held, 2005]{RueHeld:GMRFTheoryApplications}
Rue, H. and Held, L. (2005).
\newblock {\em {Gaussian Markov Random Fields: Theory and Applications (Chapman
  \& Hall/CRC Monographs on Statistics \& Applied Probability)}}.
\newblock Chapman and Hall/CRC, 1 edition.

\bibitem[Rue et~al., 2009]{Rue:INLA}
Rue, H., Martino, S., and Chopin, N. (2009).
\newblock Approximate bayesian inference for latent gaussian models by using
  integrated nested laplace approximations.
\newblock {\em Journal of the Royal Statistical Society: Series B (Statistical
  Methodology)}, 71(2):319--392.

\bibitem[Saad, 2003]{Saad:IterMethodsSparseLinearSystems}
Saad, Y. (2003).
\newblock {\em Iterative Methods for Sparse Linear Systems, Second Edition}.
\newblock Society for Industrial and Applied Mathematics, 2 edition.

\bibitem[Sain et~al., 2011]{Sain:fANOVAandRegClimateExperiments}
Sain, S.~R., Nychka, D., and Mearns, L. (2011).
\newblock {Functional ANOVA and regional climate experiments: a statistical
  analysis of dynamic downscaling}.
\newblock {\em Environmetrics}, 22(6):700--711.

\bibitem[Stein et~al., 2013]{Stein:StochApproxScoreFunctionsGPs}
Stein, M.~L., Chen, J., and Anitescu, M. (2013).
\newblock {Stochastic approximation of score functions for Gaussian processes}.
\newblock {\em The Annals of Applied Statistics}, 7(2):1162--1191.

\bibitem[Telesca and Inoue, 2008]{TelescaInoue:BayesHierarchCurveReg}
Telesca, D. and Inoue, L. Y.~T. (2008).
\newblock {Bayesian Hierarchical Curve Registration}.
\newblock {\em Journal of the American Statistical Association},
  103(481):328--339.

\end{thebibliography}

\end{document}